\documentclass{article}%
\usepackage{authblk}
\usepackage{amssymb}%
\usepackage{amsmath}
\usepackage{amsthm}
\usepackage{bbm}
\usepackage{xcolor}
\usepackage{algorithm}
\usepackage{algorithmic}
\usepackage{amsfonts}%
\usepackage{graphicx}
\usepackage{color}%
\usepackage{url}%
\usepackage{appendix}
\usepackage{comment}
\usepackage{soul}
\usepackage{subcaption}
\usepackage{pgfplots}
\usepackage{adjustbox}
\usepackage{authblk}
\usepackage[colorinlistoftodos]{todonotes}
\newcommand{\stkout}[1]{\ifmmode\text{\sout{\ensuremath{#1}}}\else\sout{#1}\fi}

\usepackage[a4paper, total={6in, 9in}]{geometry}

\newtheorem{theorem}{Theorem}[section]

\newtheorem{assumption}[theorem]{Assumption}
\newtheorem{corollary}[theorem]{Corollary}
\newtheorem{definition}[theorem]{Definition}
\newtheorem{example}[theorem]{Example}
\newtheorem{lemma}[theorem]{Lemma}
\newtheorem{proposition}[theorem]{Proposition}

\newcommand{\bb}[1]{\boldsymbol{#1}}
\newcommand{\E}[1]{\mathbb{E}[{#1}]}

\usepackage{tikz}  
  
\usetikzlibrary{shapes, calc, fit, positioning, arrows, arrows.meta,} 
\tikzset{  
    -Latex,auto,node distance =1.5 cm and 1.3 cm, thick,
    state/.style ={circle, draw, minimum width = 0.9 cm}, 
    state2/.style ={draw=white, minimum width = 0.9 cm}, 
    point/.style = {circle, draw, inner sep=0.18cm, fill, node contents={}},  
    bidirected/.style={Latex-Latex,dashed}, 
    el/.style = {inner sep=2.5pt, align=right, sloped}  
}
\providecommand{\keywords}[1]
{
  \textbf{\textit{Keywords---}} #1
}

\begin{document}

\title{Propagating moments in probabilistic graphical models with polynomial regression forms for decision support systems
}

\author[1]{Victoria Volodina}
\author[2]{Nikki Sonenberg}
\author[1]{Peter Challenor}
\author[3]{Jim Q. Smith}
\affil[1]{Department of Mathematics and Statistics, University of Exeter, Exeter, United Kingdom}
\affil[2]{Heilbronn Institute
for Mathematical Research, University of Bristol, Bristol, United Kingdom}
\affil[3]{Department of Statistics, University of Warwick, Coventry, United Kingdom}

\date{}
\maketitle

\begin{abstract}
Probabilistic graphical models are widely used to model complex systems under uncertainty. Traditionally, Gaussian directed graphical models are applied for analysis of large networks with continuous variables as they can provide conditional and marginal distributions in closed form simplifying the inferential task. The Gaussianity and linearity assumptions are often adequate, yet can lead to poor performance when dealing with some practical applications. In this paper, we model each variable in graph $\mathcal{G}$ as a polynomial regression of its parents to capture complex relationships between individual variables and with a utility function of polynomial form. We develop a message-passing algorithm to propagate information throughout the network solely using moments which enables the expected utility scores to be calculated exactly. Our propagation method scales up well and enables to perform inference in terms of a finite number of expectations. We illustrate how the proposed methodology works with examples and in an application to decision problems in energy planning and for real-time clinical decision support.

\end{abstract}


\section{Introduction}

In recent years, probabilistic graphical models have been widely used to describe complex dependencies among variables in various fields including ecology \cite{Arif2023}, biology \cite{Ni2022}, epidemiology \cite{Tennant2021}, and decision-support systems \cite{Barons2021,Shenvi2023}. 
Here a probabilistic graphical model is a directed acyclic graph (DAG) that consists of two components: variables represented by vertices in a graph $\mathcal{G}$, and dependencies by edges connecting pairs of vertices, and associated conditional probability distributions. These two components are assumed to represent our knowledge of the system and can be used to perform inference about the system \cite{Cowell2007,Lauritzen1988}. 

Within a decision-support context, a Bayesian analysis may focus on obtaining a decision maker's and multiple \emph{panel} of experts' beliefs about individual variables or propagating the effects of received evidence through the graph, coherently revising beliefs in the vertices that are still not established. To do so, fast calculation of conditional and marginal probabilities on variables is required, motivating efficient algorithms for propagating information through the DAG. These computations are significantly simplified by the conditional independence assumption embedded in the graphical model, where the joint density probability function of variables in $\mathcal{G}$ can be written as a product of the individual density functions, conditional on their parent nodes. 

However, it has been demonstrated that the evaluation of conditional and marginal probabilities can become cumbersome for large networks with discrete variables, thus \cite{Lauritzen1988} proposed efficient algorithms for transfer between representations for the joint probability distribution to perform conditioning and marginalisation. When considering continuous variables, it is common to use a Gaussian directed graphical model, where the conditional probability model for each variable is a linear Gaussian model of its parents. This assumption provides a one-to-one correspondence between a multivariate Gaussian distribution and the Gaussian directed graphical model and allows us to easily obtain the marginal and conditional distributions. 
 
Interest in learning the DAG from data goes back many years \cite{Spirtes2000}, and recently there has been an increasing interest in the learning of a Gaussian DAG from observational data using machine learning techniques \cite{Castelletti2020,Gao2022}.
The assumption of Gaussianity for all variables in the system can prove restrictive when dealing with practical problems, as regression relationships in DAG contain no interaction terms, and explanatory effects between variables is only seen through the effect of different variables on the mean of the other, not for example, their covariance.

However, the marginal posterior distributions of individual variables can become analytically intractable once we relax the Gaussianity assumption. To the best of the authors' knowledge, there are no general and closed-form methods available for propagating information in settings where the conditional probability distributions are not Gaussian. Some approximate methods have been proposed. For instance, prior to performing inferential tasks, we can perform discretisation of variables \cite{Aguilera2011} and treat them as discrete random variables. We note that the accuracy of the inference largely depends on the discretisation band, and that the computation can become impractical for large networks.  Another group of approximate methods based on simulation, where a sample of the variables in the network is obtained via Monte Carlo simulation, and then the marginal probability distributions are estimated from it \cite{Aguilera2011, Salmeron2000}. Similar to discretisation, these methods do not scale up well.

Our method is derived from the idea that message-passing algorithms do not depend upon the underlying distributions of the variables of the system \cite{Cowellbook,Jensen2007,Smith2010}, and can be written in terms of moments rather than distributions \cite{Lauritzen1992}. In this paper, we develop a new inferential method for probabilistic graphical models capable of capturing complex relationships between individual variables and their parents. In this approach, each variable in $\mathcal{G}$ is defined as a polynomial regression model of its parents. Then by applying simple algebraic techniques combined with properties of conditional expectation, we obtain the moments required to be transferred throughout the system to obtain the corresponding moments of targeted marginal distributions. This approach is distribution free, and we do not require the full mixing (convergence) associated with the full probability distribution of each component model, as compared to Markov chain Monte Carlo (MCMC). In addition, we develop a message-passing algorithm to enable efficient propagation of information throughout the system for a fast way of calculating expected utility scores for different options for the utility-maximising decision-maker.

Our inferential method is especially relevant in the decision-support context, as part of Bayesian integrating decision support systems (IDSS) \cite{Leonelli2018,Smith2015}, where each vertex represents an individual component model developed by a \emph{panel} of experts. Here we assume all the panels agree on a polynomial form of utility function. This assumption is made for two reasons. Firstly, its decomposability allows us to consider each term separately and therefore simplifies the message-passing process. Secondly, when the utility function is only approximately polynomial, we can consider only the first few terms in the Maclaurin expansion as long as the approximating error term is not too large. In this case, the focus is on propagating information forward to the margins of the attributes, the random vector of arguments of a utility function, to compute the expected utilities.

The article has the following structure. In Section \ref{sec:background} we introduce notation and the mathematical preliminaries for a single panel in the network model.  We define and consider different classes of regression models in Section \ref{sec:graphical} and comment on the complexity of a single panel. In Section \ref{sec:momentsuff} we present a message-passing algorithm for the calculation of the expected utility and provide formulas for the required moment computations. In Section \ref{sec:application} we apply our method to two applications, energy planning and a healthcare setting where planning and treatment decisions are required. Section \ref{sec:discussion} contains concluding remarks.

\section{Polynomial regression model} \label{sec:background}
\subsection{Preliminaries}
Since the focus of this paper is to develop an inferential methodology in the context of decision support, we start by defining $m$ models developed by panels $\left\{ Q_{1},Q_{2},\dots,Q_{m}\right\} $ indexed by the partial order generated by a graph $\mathcal{G}$ and letting $Q_{m+1}$ represent the decision maker. We assume that panels produce outputs 
$\bb{Y} = \{Y_{1}, Y_{2}, \dots, Y_{m} \}$ 
and can receive inputs  $\boldsymbol{Y}_{\Pi_{j}}$, where
 $\Pi_j$ is the set of indices such that $\boldsymbol{Y}_{\Pi_j} \subseteq \{Y_{1},Y_{2},\dots,Y_{j-1}\}$,
 where each variable in the set $\boldsymbol{Y}_{\Pi_{j}}$ is a parent of $Y_j$. For each decision $d\in\mathcal{D}$, we define the utility function that takes the form 
\begin{equation*}
    U(\bb{Y}, d)=\sum_{i=1}^K k_iU_i(\bb{Y}_{H_i}, d),
\end{equation*}
where $k_i\in (0, 1)$ are the criterion weights agreed by all panels and $U_i(\bb{Y}_{H_i}, d)$ is the function of $\bb{Y}_{H_i}$ with indices $H_i\subseteq \{1, 2, \dots, m \}$ which can contain polynomial terms and interactions. 
To inform the decision-making, we need to provice the set of expected utility (EU) scores $\{\overline{U}(d): d\in\mathcal{D}\}$, so that 
we can then recommend
$ d^{\ast}=\arg\max_{d\in\mathcal{D}}\overline{U}(d)$, where the expected utility is given by
\begin{align*}\label{eqn_util1}
\overline{U}(d)=\int \overline{U}(d|\bb{\theta}) \pi(\bb{\theta}|d) d\bb{\theta},
\end{align*}
and the conditional expected utility is
\begin{equation*}
    \overline{U}(d|\bb{\theta})=\int U(\bb{y}, d)f(\bb{y}|\bb{\theta}, d) d\bb{y}.
\end{equation*}
If $U(\bb{Y},d)$ is linear in $\bb{Y}$,  to find $d^{\ast}$ only a vector of means is required for each $d$. 
If $U(\cdot,d)$ is polynomial of degree $n$, it can be shown that only $n^\text{th}$ moments of key output variables of $f(\bb{y}|\bb{\theta}, d)$ are required. If $U(\cdot,d)$ is nonlinear, then the expectations are computed numerically. In this paper, we assume the utility function has a polynomial form.

Following conditional independence assumptions, any Bayesian network (BN) $\mathcal{G}$ over continuous random variables $\bb{Y}$ can be represented in terms of a factorisation of conditional probability densities, given by the form
\begin{equation}
\label{eq:py_dynamic}
  f(\bb{y}|\bb{\theta}, d)=\prod_{j=1}^m f_j(y_j|\bb{y}_{\Pi_j}, \bb{\theta}_j, d).
\end{equation}
We note that the graph $\mathcal{G}$ and the joint probability density $f(\bb{y}|\bb{\theta}, d)$ define the network model. In order to specify the family of joint distributions of $\bb{Y}$ consistent with $\mathcal{G}$, we need to provide candidate distributions of $  [Y_j|\bb{Y}_{\Pi_j}=\bb{y}_{\Pi_j} ]$, for $j=1, 2, \dots, m$, that can be of any form. This makes the proposed framework generic. 

We proceed to define the models using monomial notation \cite{Cox2005,Pistone2000}. Each $Y_j|\bb{Y}_{\Pi_j}$ can be written as a regression model
\begin{equation}
    Y_j|\bb{Y}_{\Pi_j}=\sum_{i=1}^{\infty} \theta_{ji}\phi_{ji}(\bb{Y}_{\Pi_j}) + v_j,
\end{equation}
where $\{\theta_{ji}, i =1, 2, \dots \}$ is the set of parameter vectors, 
 $\{\phi_{ji}(\bb{Y}_{\Pi_j}), i=1, 2, \dots \}$ is a countably infinite basis, which in this paper is considered to be polynomial (another example is the Fourier basis),  and $v_j$ is an observation error independent from $\bb{Y}_{\Pi_j}$ and regression coefficients. A polynomial regression model for $Y_j$ can be written as
 \begin{align}  
  \label{eq:polynomial_reg}  
Y_{j}|\bb{Y}_{\Pi_j}&=\sum_{\bb{a}_j\in A_{j}}\theta_{j\bb{a}_j%
}\boldsymbol{Y}_{\Pi_j}^{\bb{a}_j}+v_{j}\text{ with monomial }  \boldsymbol{Y}_{\Pi_j}^{\bb{a}_{j}} = \prod_{i\in \Pi_j}Y_{i}^{a_{ji}},
\end{align}
where $A_{j}$ is the set of the exponent vectors on the monomial covariates of $Y_j$. The terms in the monomial are ordered lexicographically in $\bb{a}_j$ according to the order generated by $\mathcal{G}$. We have the following two assumptions for the models.
\begin{assumption}
\label{assume:independence}
A joint prior is specified over $\{(\bb{\theta}_j, v_j), j=1, 2, \dots, m \}$, where $\bb{\theta}_{j}\triangleq\{ \theta_{j\bb{a}_j},\bb{a}_j\in A_j\}$, and that 
$\amalg_{j=1}^{m}v_{j}$ and
$\amalg_{j=1}^{m}\boldsymbol{\theta}_{j}$. 
\end{assumption}
\begin{assumption}\label{assume:moments}
    The moment-generating function of $\boldsymbol{Y}$ exists.
\end{assumption}
Under Assumption \ref{assume:moments}, we can compute the moments of $Y_j^b|\bb{Y}_{\Pi_j}$, which are uniquely defined in terms of the non-central moments of $\bb{Y}_{\Pi_j}^{\bb{a}_j}$, regression coefficients $\theta_{j\bb{a}_j}$ and 
the higher-order non-central moments of $v_j$, $j=1, 2, \dots, n$. 
We note that the exponent vector alone identifies a monomial, that is, the monomial $\bb{Y}_{\Pi_j}^{\bb{a}_j}$ can be represented  by $\boldsymbol{a}_j$.

 \begin{example}\label{ex:ex1}
    Consider the polynomial regression model for $Y_4$ as 
    $$Y_4 = \theta_{4,000}+ \theta_{4,100}Y_{1}+ \theta_{4,120}Y_{1} Y_{2}^2  + \theta_{4,002}Y_{3}^2 +v_4,$$
then $A_4 = \{(0,0,0),(1,0,0),$ $(1,2,0), (0,0,2)\}$ and $\boldsymbol{\theta}_{4\boldsymbol{a}_4}= (\theta_{4,000}, \theta_{4, 100}, \theta_{4, 120}, \theta_{4, 002})$. 
\end{example}
For the computation of the interaction terms, for $b \in \mathbb{Z}_{\ge 0}$  and some exponent vector $\boldsymbol{g}$, we consider
\begin{equation}\label{monoYX}
Y_j^b\bb{Y}_{\Pi_j}^{\bb{g}}=\big( \sum_{\bb{a}_j\in A_j}\theta_{j\bb{a}_j}\bb{Y}_{\Pi_j}^{\bb{a}_j} + v_j\Big)^b\bb{Y}_{\Pi_j}^{\bb{g}}.
\end{equation}
Further, for the expansion of the expression within the brackets, define $A_j^b$ as the exponent set of $Y_j^b$, then
\begin{align}\label{defn_Abj}
    A_j^b=\left\{  \boldsymbol{a}:  \boldsymbol{a} = \boldsymbol{a}_i+\boldsymbol{a}%
_j+\cdots+\boldsymbol{a}_k,\text{where }\boldsymbol{a}_i,\boldsymbol{a}%
_j,\ldots,\boldsymbol{a}_k\in A_{j} \text{ and } S_n=b\right\},
\end{align}
where $S_n$ denotes the number of terms in the summation. Thus for a model with $\ell$ monomial terms to the power of $b$, this will have $\binom{\ell+b-1}{b}$ elements. As the multiplication of monomials is achieved by adding exponents, we also define 
\begin{align}
A_j^b+\left\{  \boldsymbol{g}\right\} =\left\{  \boldsymbol{a}%
:\boldsymbol{a}=\boldsymbol{a}^{\prime}+\boldsymbol{g},\text{where
}\boldsymbol{a}^{\prime}\in A_j^b\right\}.
\end{align} 

\begin{example} Consider the polynomial regression model for $Y_4$ in Example \ref{ex:ex1}. By Equation (\ref{defn_Abj}), we can write down the set of exponents required for the calculation of $Y_4^2$ as listed in Table \ref{tab:exponents_A4}.

\begin{table}[h!]
\centering
\caption{Set of exponents $A_4^2$}
\label{tab:exponents_A4}
 \begin{tabular}{|c c |} 
 \hline
 $\bb{a}_1+\bb{a}_1=(0, 0, 0)$ & $\bb{a}_2+\bb{a}_3=(2, 2, 0)$ \\ 
$\bb{a}_1 +\bb{a}_2=(1, 0, 0)$ & $\bb{a}_2+\bb{a}_4=(1, 0, 2)$ \\
$\bb{a}_1+\bb{a}_3=(1, 2, 0)$ & $\bb{a}_3+\bb{a}_3=(2, 4, 0)$  \\
$\bb{a}_1+\bb{a}_4=(0, 0, 2)$ & $\bb{a}_3+\bb{a}_4 = (1, 2, 2)$  \\
$\bb{a}_2+\bb{a}_2 = (2, 0, 0)$ & $\bb{a}_4 +\bb{a}_4 = (0, 0, 4)$  \\ [1ex] 
 \hline
 \end{tabular}
\end{table}
\end{example}

  
 \subsection{Single panel calculations} \label{subsec:singlepanel}

 For a single panel overseeing the distribution of one of the response variables, we next demonstrate how to calculate non-central moments and present formulae for the predictive moments of the dependent variable and its associated regressors.
 
 Let $\mu_{j\bb{a}}$ denote $\mathbb{E}[\bb{Y}_{\Pi_j}^{\bb{a}}]$ and  $m_{i}(v_j)= \mathbb{E}\left(  v_{j}^{i}\right)$ specifies the non-central moments of the observation error, 
and we assume that $m_{1}(v_{j})=0$ and $m_{2}(v_{j})=\mathbb{V}(v_{j})=V_j$. Using this notation, we demonstrate how the required moments are computed in order for the decision-maker to calculate the utility scores.
 
\begin{proposition}
\label{th:polynomialder}
For a regression model $Y_j$, suppose $\bb{\theta}_j,v_{j}$ and $\boldsymbol{Y}_{\Pi_j}$ are all mutually
independent, for some integer $b$ and exponent vector $\bb{g}$,
\begin{equation}\label{eq:mod_prop27}
    \mathbb{E}[Y_j^b\bb{Y}_{\Pi_j}^{\bb{g}}] = \sum_{k=0}^b\binom{b}{k}m_k(v_j)\sum_{\substack{\bb{a}\in A_j^{b-k}\\ t_{\bb{a}_1} + t_{\bb{a}_2} + \dots +t_{\bb{a}_{|A_j|}} = b-k\\ t_{\bb{a}_1}, t_{\bb{a}_2}, \dots, t_{\bb{a}_{|A_j|}}\geq 0}}\binom{b-k}{t_{\bb{a}_1}, t_{\bb{a}_2}, \dots t_{\bb{a}_{|A_j|}}}\mathbb{E}\Big(\prod_{\bb{a}_i\in A_j, \bb{a}_i+\bb{a}_j+\dots +\bb{a}_m = \bb{a}}\theta_{j\bb{a}_i}^{t_{\bb{a}_i}} \Big)\mu_{j, \bb{a}'},
\end{equation}
with $
\bb{a}'=\bb{a} + \bb{g}$
and 
$$
\binom{b-k}{t_{\bb{a}_1}, t_{\bb{a}_2}, \dots t_{\bb{a}_{|A_j|}}} = \frac{(b-k)!}{t_{\bb{a}_1}! t_{\bb{a}_2}! \dots t_{\bb{a}_{|A_j|}}!}.
$$
\end{proposition}

\begin{proof}
Rearranging Equation \eqref{monoYX} with $v_j$ as
the argument, then taking expectations, we have
\begin{align*}
    \mathbb{E}[Y_j^b\bb{Y}_{\Pi_j}^{\bb{g}}]&=\mathbb{E}\Bigg[\Big(\sum_{\bb{a}_j\in A_j}\theta_{j\bb{a}_j}\bb{Y}_{\Pi_j}^{\bb{a}_j}+v_j \Big)^b\bb{Y}_{\Pi_j}^{\bb{g}} \Bigg]=\sum_{k=0}^b\mathbb{E}\Bigg[\binom{b}{k}v_j^k\Big(\sum_{\bb{a}_j\in A_j}\theta_{j\bb{a}_j}\bb{Y}_{\Pi_j}^{\bb{a}_j} \Big)^{b-k} \bb{Y}_{\Pi_j}^{\bb{g}}\Bigg]
\end{align*}
Since $v_{j}\amalg\left(  \boldsymbol{\theta}_j,\bb{Y}_{\Pi_j}\right)$, 
we can write
$$
\mathbb{E}[Y_j^b\bb{Y}_{\Pi_j}^{\bb{g}}] = \sum_{k=0}^b\binom{b}{k}m_k(v_j)\mathbb{E}\Bigg[\Big(\sum_{\bb{a}_j\in A_j}\theta_{j\bb{a}_j}\bb{Y}_{\Pi_j}^{\bb{a}_j} \Big)^{b-k} \bb{Y}_{\Pi_j}^{\bb{g}} \Bigg].
$$
By the multinomial theorem \cite{Spiegel1990}
we have
$$
\Big(\sum_{\bb{a}_j\in A_j}\theta_{j\bb{a}_j}\bb{Y}_{\Pi_j}^{\bb{a}_j}\Big)^{b-k}=\sum_{\substack{\bb{a}\in A_j^{b-k}\\ t_{\bb{a}_1} + t_{\bb{a}_2} + \dots +t_{\bb{a}_{|A_j|}} = b-k\\ t_{\bb{a}_1}, t_{\bb{a}_2}, \dots, t_{\bb{a}_{|A_j|}}\geq 0}}\binom{b-k}{t_{\bb{a}_1}, t_{\bb{a}_2}, \dots t_{\bb{a}_{|A_j|}}}\Big(\prod_{\bb{a}_i\in A_j, \bb{a}_i+\bb{a}_j+\dots +\bb{a}_m = \bb{a}}\theta_{j\bb{a}_i}^{t_{\bb{a}_i}} \Big)\bb{Y}_{\Pi_j}^{\bb{a}}.
$$
Collecting the terms gives
\begin{align*}
\mathbb{E}[Y_j^b\bb{Y}_{\Pi_j}^{\bb{g}}] &= \sum_{k=0}^b\binom{b}{k}m_k(v_j)\sum_{\substack{\bb{a}\in A_j^{b-k}\\ t_{\bb{a}_1} + t_{\bb{a}_2} + \dots +t_{\bb{a}_{|A_j|}} = b-k\\ t_{\bb{a}_1}, t_{\bb{a}_2}, \dots, t_{\bb{a}_{|A_j|}}\geq 0}}\binom{b-k}{t_{\bb{a}_1}, t_{\bb{a}_2}, \dots t_{\bb{a}_{|A_j|}}}\mathbb{E}\Big(\prod_{\bb{a}_i\in A_j, \bb{a}_i+\bb{a}_j+\dots +\bb{a}_m = \bb{a}}\theta_{j\bb{a}_i}^{t_{\bb{a}_i}} \Big)\mathbb{E}[\bb{Y}_{\Pi_j}^{\bb{a}+\bb{g}}]
\end{align*}
\end{proof}

\begin{example} (Example \ref{ex:ex1} continued)
   To compute $\mathbb{E}\big(Y_4^b\bb{Y}_{\Pi_4}^{\bb{g}}\big)$, where $b=2$ and $\bb{g}=(1, 2, 3)$, by Proposition \ref{th:polynomialder},
    \begin{align*}
\mathbb{E}\big(Y_4^b\bb{Y}_{\Pi_4}^{\bb{g}}\big)&=\binom{2}{0}m_0(v_4)\sum_{\substack{\bb{a}\in A_4^2\\ t_{\bb{a}_1} + t_{\bb{a}_2} + \dots +t_{\bb{a}_{|A_4|}} = 2\\ t_{\bb{a}_1}, t_{\bb{a}_2}, \dots, t_{\bb{a}_{|A_4|}}\geq 0}}\binom{2}{t_{\bb{a}_1}, t_{\bb{a}_2}, \dots t_{\bb{a}_{|A_4|}}}\mathbb{E}\Big(\prod_{\bb{a}_i\in A_4, \bb{a}_i+\bb{a}_j+\dots +\bb{a}_m = \bb{a}}\theta_{4\bb{a}_i}^{t_{\bb{a}_i}} \Big)\mu_{4, \bb{a}+\bb{g}}\\
        &+\binom{2}{1}m_1(v_4)\sum_{\substack{\bb{a}\in A_4\\ t_{\bb{a}_1} + t_{\bb{a}_2} + \dots +t_{\bb{a}_{|A_4|}} = 1\\ t_{\bb{a}_1}, t_{\bb{a}_2}, \dots, t_{\bb{a}_{|A_4|}}\geq 0}}\binom{1}{t_{\bb{a}_1}, t_{\bb{a}_2}, \dots t_{\bb{a}_{|A_j|}}}\mathbb{E}\Big(\prod_{\bb{a}_i\in A_4, \bb{a}_i+\bb{a}_j+\dots +\bb{a}_m = \bb{a}}\theta_{4\bb{a}_i}^{t_{\bb{a}_i}} \Big)\mu_{4, \bb{a}+\bb{g}}\\
        &+m_2(v_4).
    \end{align*}
Consider the first term of the expression
    $$
    A_4^2=\big\{(0,0,0), (1,0,0), (1,2,0), (0,0,2), (2, 0, 0), (2, 2, 0), (1, 0, 2), (2, 4, 0), (1, 2, 2), (0, 0, 4)\big\},
    $$
    then we can deduce that that $\sum_{\bb{a}\in A_4^2}\binom{2}{t_{\bb{a}_1}, t_{\bb{a}_2}, \dots t_{\bb{a}_{|A_4|}}}\mathbb{E}\Big(\prod_{\bb{a}_i\in A_4, \bb{a}_i+\bb{a}_j+\dots +\bb{a}_m = \bb{a}}\theta_{4\bb{a}_i}^{t_{\bb{a}_i}} \Big)\mu_{4, \bb{a}+\bb{g}}$ can be written as
    \begin{align*}
    \mathbb{E}[\theta_{4,000}^2]\mu_{4,123} &+ 2\mathbb{E}[\theta_{4,000}\theta_{4,100}]\mu_{223}+2\mathbb{E}[\theta_{4,000}\theta_{4,120}]\mu_{4,243}+2\mathbb{E}[\theta_{4,000}\theta_{4,002}]\mu_{4,125}\\
    &+\mathbb{E}[\theta_{4,100}^2]\mu_{4,323}+2\mathbb{E}[\theta_{4,100}\theta_{4,120}]\mu_{4,343} +2\mathbb{E}[\theta_{4,100}\theta_{4,002}]\mu_{4,225}\\
    &+\mathbb{E}[\theta_{4,120}^2]\mu_{4,363}+2\mathbb{E}[\theta_{4,120}\theta_{4,002}]\mu_{4,345}+\mathbb{E}[\theta_{4,002}^2]\mu_{4,127}.
    \end{align*}
    Similarly, we can derive that
    $$
    A_4=\{(0,0,0),(1,0,0),(1,2,0)(0,0,2)\}
    $$
   and that $\sum_{\bb{a}\in A_4}\binom{1}{t_{\bb{a}_1}, t_{\bb{a}_2}, \dots t_{\bb{a}_{|A_j|}}}\mathbb{E}\Big(\prod_{\bb{a}_i\in A_4, \bb{a}_i+\bb{a}_j+\dots +\bb{a}_m = \bb{a}}\theta_{4\bb{a}_i}^{t_{\bb{a}_i}} \Big)\mu_{4, \bb{a}+\bb{g}}$
   can be written as
    $$
\mathbb{E}[\theta_{4,000}]\mu_{4,123}+\mathbb{E}[\theta_{4,100}]\mu_{4,223}+\mathbb{E}[\theta_{4,120}]\mu_{4,243}+\mathbb{E}[\theta_{4,002}]\mu_{4,125}.
    $$
    Finally, since $m_1(v_4)=0$ and $m_0(v_4)=1$, we have
    \begin{align*}
\mathbb{E}\big(Y_4^b\bb{Y}_{\Pi_4}^{\bb{g}}\big)&=\mathbb{E}[\theta_{4,000}^2]\mu_{4,123} + 2\mathbb{E}[\theta_{4,000}\theta_{4,100}]\mu_{4,223}+2\mathbb{E}[\theta_{4,000}\theta_{4,120}]\mu_{4,243}+2\mathbb{E}[\theta_{4,000}\theta_{4,002}]\mu_{4,125}\\
    &+\mathbb{E}[\theta_{4,100}^2]\mu_{4,323}+2\mathbb{E}[\theta_{4,100}\theta_{4,120}]\mu_{4,343} +2\mathbb{E}[\theta_{4,100}\theta_{4,002}]\mu_{4,225}\\
    &+\mathbb{E}[\theta_{4,120}^2]\mu_{4,363}+2\mathbb{E}[\theta_{4,120}\theta_{4,002}]\mu_{4,345}+\mathbb{E}[\theta_{4,002}^2]\mu_{4,127} + V_4.
    \end{align*} 
\end{example}
It can be checked when computing $\mathbb{E}[Y_j\bb{Y}_{\Pi_j}^{\bb{g}}]$, the panel $Q_j$ only need to consider $\mu_{j\bb{a}'}$ with $\bb{a}'=\bb{a}_j+\bb{g}$ and posterior expectations of its regression coefficients, $\mathbb{E}[\theta_{j\bb{a}_j}]$, for $\bb{a}_j\in A_j$ in the calculations. But in the case we are interested in computing $\mathbb{E}[Y_j^2\bb{Y}_{\Pi_j}^{\bb{g}}]$, we have to consider $\theta_{j\bb{a}_j}^2$, even if the distribution of regression parameters is, for example, Gaussian, the distribution of $\theta_{j\bb{a}_j}^2$ could be unfamiliar. 
However, within the proposed framework, we only need to obtain $\mathbb{E}[\theta_{j\bb{a}_j}^2]$, which is much simpler to calculate. We note that the proposed method is Bayesian since we include the estimated uncertainty of parameters in our calculations in Equation (\ref{eq:mod_prop27}) contrary to the frequentist approach where parameters are fixed at maximum likelihood estimated values. In many cases, to compute the attributes we will require various higher-order joint moments of the system, which capture a panel's uncertainty about its estimates.

 \section{Classes of polynomial regression models}\label{sec:graphical}

Using the polynomial regression model of Section \ref{sec:background}, we  define classes of these models with specific forms of a set of exponent vectors, building on classical ideas of graphical structures \cite{Lauritzen1988}.

\begin{definition}
\label{defn:extreme_points}
For a set of exponent vectors $A_j$ obtained from regression model $Y_j$, let $A_j^{\ast}$ denote the set of  \textbf{corner points}  where
\begin{align}
A_j^{\ast}=\left\{  \boldsymbol{a}^{\ast}\in A_j:\nexists\text{
}\boldsymbol{a}\neq\boldsymbol{a}^{\ast}\in A_j\text{ such that }\boldsymbol{a}%
^{\ast}\leq\boldsymbol{a}\right\},
\end{align}
where $\boldsymbol{a}\leq\boldsymbol{a}^{\ast}$ under product order
\footnote{The \emph{product order} for a set $A\neq \emptyset$: Given two pairs $\bb{a}=(a(1), a(2), \dots, a(n))$ and $\bb{a}'=(a'(1), a'(2), \dots, a'(n))$, where $\bb{a}, \bb{a}'\in A$, we declare that $\bb{a}\leq \bb{a}'$ if and only if $a(i)\leq a'(i), i=1, 2, \dots, n$. 
}, 
i.e., $0\leq a(i)\leq a^{\ast}(i)$, for all $i=1,2,\ldots,n.$ 
Further, define the closure of set $A_j^{\ast}$ as
\begin{align}
\bar{A}_j^{\ast}=\left\{  \boldsymbol{a}:\boldsymbol{a}%
\leq\boldsymbol{a}^{\ast},\boldsymbol{a}^{\ast}\in A_j^{\ast}
\right\}.
\end{align}
\end{definition}

 \color{black}
\begin{example}
A regression model for $Y_3$ defined as 
\begin{equation*}
       Y_3 = \theta_{3,00}+ \theta_{3,10}Y_1+\theta_{3,20}Y_1^2+ \theta_{3,11}Y_1Y_2+ \theta_{3,12}Y_1Y_2^2+ \theta_{3,01}Y_2+ \theta_{3,02}Y_2^2+ \theta_{3,03}Y_2^3 + v_3,
\end{equation*}
has the set of corner points $$A_3^* =\{ (2, 0), (1, 2), (0, 3)\},$$ with closure  $$\bar{A}_3^{\ast}=\{(0, 0), (1, 0), (2, 0), (1, 1), (1, 2), (0, 1), (0, 2), (0, 3)\}.$$ 
\end{example}

\begin{definition}\label{def:types_of_models}
For a regression model $Y_j$, with exponent set $A_j$ corresponding to the monomial covariates, then
 \begin{itemize}
    \item $Y_j$ is a \textbf{full model} if it contains all lower-order terms, i.e., $A_j=\bar{A}_j^*$;
    \item $Y_j$ is a \textbf{simple model} if it is a full model with a single
corner point, i.e., $A_j=\bar{A}_j^*$ and $\#(A_j^{\ast})=1$;
    \item $Y_j$ is a \textbf{hierarchical model} if it is a full model and none of its corner points has an exponent greater
than one; 
    \item $Y_j$ is a \textbf{graphical regression model}  if it is a hierarchical model and there is exactly one corner point for each of the cliques  
    of its undirected graph $\mathcal{U}$. 
The undirected graph contains an edge between covariates, iff there exists a monomial term in the regression model containing both.
\end{itemize}
\end{definition}

Denote by $\bb{Y}_{\mathcal{I}_j}$ where
$\mathcal{I}_j\subseteq\Pi_j$ is a set of indices corresponding to a clique in $\mathcal{U}$. Note that a simple model and a hierarchical model are special cases of a full model. By definition, $A_j\subseteq\bar{A}_j^{\ast}$ with equality when $Y_j$ is a full model and  where $\bar{A}_j^{\ast}$ is generated by a set of corner points $A_j^*$. In this sense, the corner points $A^*_j$ form a minimal generator of $A_j$ and thus a natural index of the associated full model $A_j$.

\begin{example}
\label{ex:diff_regressions}
(i) The polynomial model for $Y_2$ on a single variable,
\begin{align*}
Y_2 = \theta_{2, 0} + \theta_{2, 1}Y_1+ \theta_{2, 2}Y_1^{2}+ \dots + \theta_{2,b}Y_1^{b} + v_2
\end{align*}
is a simple model with $A_2^*=\{b \}$, but is not hierarchical when $b>1$. 

\noindent (ii) The model for $Y_5$ on four variables,
\begin{align*}
Y_5 &= \theta_{5, 0000}+ \theta_{5, 1000}Y_{1}+\theta_{5,0100}Y_{2}+\theta_{5, 0010}Y_{3}+\theta_{5, 0001}Y_{4}\\&+\theta_{5, 1100}Y_{1}Y_{2} + \theta_{5, 1010}Y_{1}Y_{3}+\theta_{5, 0110}Y_{2}Y_{3}+\theta_{5,0011}Y_{3}Y_{4} + v_5
\end{align*}
is hierarchical but not a simple or graphical regression model. It can be seen from $\mathcal{U}$:
\[%
\begin{array}
[c]{ccccc}%
Y_1 & - & Y_3 & - & Y_4\\
| & \diagup &  &  & \\
Y_2 &  &  &  &
\end{array}
\]
defined by the  monomial terms we have edges: 
$E=\left\{ \{Y_1, Y_2 \}, \{Y_1, Y_3 \}, \{Y_2, Y_3 \}, \{ Y_3Y_4\}\right\}$. However, the model is not a graphical since the monomial associated with the clique $\bb{Y}_{\{1,2,3 \}}$ is not part of the corner points: $$A_5^*=\{(1, 1, 0, 0), (1, 0, 1, 0), (0, 1, 1, 0), (0, 0, 1, 1) \}.$$

\noindent (iii) The  model for $Y_5$ on four variables:
\begin{align*}
Y_5 &= \theta_{5,0000}+  \theta_{5,1000}Y_{1}+ \theta_{5,0100}Y_{2}+ \theta_{5,0010}Y_{3}+ \theta_{5,0001}Y_{4}+ \theta_{5,1100}Y_{1}Y_{2} +  \theta_{5,1010}Y_{1}Y_{3}\\&+ \theta_{5,0110}Y_{2}Y_{3}+ \theta_{5,0011}Y_{3}Y_{4} +  \theta_{5,1110}Y_{1}Y_{2}Y_{3} + v_5,
\end{align*}
is a graphical regression model since the set of corner points $A_5^*=\big\{ (1, 1, 1, 0), (0, 0, 1, 1) \big\}$ correspond to the cliques $\bb{Y}_{\{ 1,2,3\}}$ and $\bb{Y}_{\{3,4\}}$ in its $\mathcal{U}$.

\end{example}
 The connection between the model classes in Definition \ref{def:types_of_models} and existing work is briefly discussed. Hierarchical log-linear models introduced by Edwards \cite{Edwards2012} to analyse multivariate discrete data are similar to our model classes. If a term in the model is set to zero, then all its higher-order related terms are also set to zero. 
Graphical models are the subclass of Edward's hierarchical log-linear models, where each clique of  $\mathcal{U}$ that represents the log-linear model has an interaction term containing all the indices in that clique. We note that the full model respects the principle of marginality introduced by Nedler \cite{Nedler1977}.

\subsection{Complexity of a single panel 
}\label{subsec:complexity_panel}
When considering a polynomial model associated with a single panel, we define the complexity as the number of monomial terms that need to be generated, and proceed to comment on the complexity of the model classes in Definition \ref{def:types_of_models}. 

In a full, simple, or hierarchical model over $n$ covariates of $Y_j$, the maximum number of terms in $\bar{A}_j^*$ that need to be considered is bounded by 
\begin{equation}
\max_{\boldsymbol{a}^{\ast}\in A_j^{\ast}}%
{\displaystyle\prod\limits_{i=1}^{n}}
\left(  a^{\ast}(i)+1\right)  \leq\#\left(  \bar{A}_j^*\right)  \leq\sum_{\boldsymbol{a}%
^{\ast}\in A_j^{\ast}}%
{\displaystyle\prod\limits_{i=1}^{n}}
\left(  a^{\ast}(i)+1\right).
\label{inequality_size_from_hull}
\end{equation}
\begin{example} Consider a full model for $Y_3$ with corner points 
$
A_3^*=\{(2, 1, 0), (0, 3, 0), (0, 0, 1) \}$, then by Equation (\ref{inequality_size_from_hull}) the number of terms will be
  $  6\leq \#(\bar{A}_3^*)\leq 12.
  $
\end{example}

In a graphical regression model, each corner point is the function of regressors within its cliques in $\mathcal{U}$. Thus the cliques of interest will be the parent set of particular regressors in  $\mathcal{G}$ and uniquely determined by $\mathcal{G}$.

The graphical regression model for $Y_j$ with cliques $\bb{Y}_{\mathcal{I}_j}$ where
$\mathcal{I}_j\subseteq\Pi_j$,
contains all simple joint interaction terms up to order $\#(\mathcal{I}_j)$. Then the associated set of exponents
can be written as
\begin{align}
A_j=\left\{  \boldsymbol{a}=\left(  a(1),a(2),\ldots,a(j-1)\right), \text{where }0\leq
a(i)\leq \mathbbm{1}\{i\in\mathcal{I}_j \}\text{ for }i=1,2,\ldots, j-1\right\},
\end{align}
with $\mathbbm{1}\{i\in\mathcal{I}_j\} = 1$ if $i\in \mathcal{I}_j$,  and $\mathbbm{1}\{i\in\mathcal{I}_j\} = 0$ if $i\notin \mathcal{I}_j$.

Similarly, the set of corner points for a graphical regression model for $Y_j$ is
\begin{equation}
    A_j^*=\left\{ \bb{a}^*=\left(a^*(1),a^*(2),\ldots,a^*(j-1)\right), \text{where }a^{\ast}(i)=\mathbbm{1}\{i\in\mathcal{I}_j \}\text{ for }i=1,2,\ldots, j-1\right\},
\end{equation}
and for this case
\begin{equation}
    \#\left(  \bar{A}_j^*\right)\leq 1+\sum_{\mathcal{I}_j\subseteq V} \sum_{k=1}^{\#(\mathcal{I}_j)} \binom{\#(\mathcal{I}_j)}{k},
\end{equation}
where $V$ is the set of vertices of $\mathcal{U}$.

\begin{example}
    Consider the graphical regression model in Example \ref{ex:diff_regressions} (iii). The number of terms related to the clique $\bb{Y}_{\{ 1,2,3\}}$ is  $\sum_{k=1}^{3} \binom{3}{k} = 7$,  and for  $\bb{Y}_{\{3,4\}}$ is $          \sum_{k=1}^{2} \binom{2}{k}  = 3. $  Then total number of exponents in the closure set is $\#(\bar{A}_5^*)\leq 11$.
\end{example}

\subsection{Operations with exponent sets}

As the result of operations on regression model $Y_j$ considered in Section \ref{subsec:singlepanel}, we proceed to derive the set of corner points and the closure of the corresponding obtained set.

We introduce the sets of exponent vectors $A$ and $G$ with the corresponding sets of corner points $A^{\ast}$ and $G^{\ast}$, respectively.
For $b\in\mathbb{Z}^+$ define 
\begin{align}
H=G+bA=\left\{  \boldsymbol{h}:\boldsymbol{h=g}+b\boldsymbol{a}\text{
for some }\boldsymbol{g}\in G,\text{ }\boldsymbol{a}\in A\right\}.
\end{align}
\begin{lemma}
The corner points $H_{I}^{\ast}$of $H_{I}=
{\cup_{i\in I}}
b_{i}A$ are the corner points of $b^{\ast}A$ where
$b^{\ast}=\max_{i\in I}b_{i}$.

\end{lemma}

\begin{lemma}
\label{lemma:corner_points1}
The corner points $H^{\ast}$ of $H=G+bA$ are of the form
$\boldsymbol{h}^{\ast}=\boldsymbol{g}^{\ast}+b\boldsymbol{a}^{\ast}$ where
$\boldsymbol{g}^{\ast}\in G^{\ast}$ and $\boldsymbol{a}^{\ast}\in A^{\ast}.$
\end{lemma}

\begin{proof}
By contradiction, if $\boldsymbol{h}^{\ast}$ were not of this form then
either $\boldsymbol{g}^{\ast}$or $\boldsymbol{a}^{\ast}$ would be strictly
dominated by another element in $G$ or $A$. In either case this would mean
by definition of $G+bA$ that $\boldsymbol{h}^{\ast}$ was strictly
dominated in $G+bA$.
\end{proof}
These results imply that by constructing sets of type $H_{I}=\cup_{i\in I}\{  G+b_{i}A\} $ we  only need to calculate the corner points of $G$ and $b^{\ast}A$. Following the results in Section \ref{subsec:singlepanel}, for a polynomial model $Y_j$, we are interested in deriving the new exponent sets $A_j^b$ and $(A_j^b)^*$, for some $b$. We demonstrate that the set of corner points of exponent set of $Y_j^b$ can be derived from the set of corner points of $Y_j$.
 
\begin{lemma}\label{lemma_simple} 
For a regression model $Y_j$, with a set of exponents, $A_j$, and a set of corner points, $A_j^*$, for  $b\in\mathbb{Z}^+$, we derive that 
\begin{equation}\label{eq:powered_set_lemma}
    (A_j^b)^*=\Big\{\bb{a}:\bb{a}=\bb{a}_i + \bb{a}_j+\dots+\bb{a}_k,\text{ where } \bb{a}_i, \bb{a}_j, \dots, \bb{a}_k\in A_j^*\text{ and }S_n=b\Big\}, 
\end{equation}
then $(A_j^b)^*=(A_j^*)^b$. Similarly, 
\begin{equation}\label{eq:powered_set_lemma2}
    \overline{(A_j^b)^*} = \Big\{\bb{a}:\bb{a}=\bb{a}_i+\bb{a}_j+\cdots + \bb{a}_k, \text{ where }\bb{a}_i, \bb{a}_j, \cdots \bb{a}_k\in \bar{A}_j^*\text{ and }S_n=b  \Big\},
\end{equation}
then $\overline{(A_j^b)^*}=(\bar{A}_j^*)^b$.
\end{lemma}
 
\begin{proof}

By contradiction, if $\bb{a}$ were not of this form then terms $\bb{a}_i, \bb{a}_j\cdots, \bb{a}_k$ would be strictly dominated by other terms in $A_j$. Therefore by definition of $A_j^b$ that would mean that $\bb{a}$ is strictly dominated by terms in $A_j^b$. Same holds for the closure set of the corner points $(A_j^b)^*$.
\end{proof}
\color{black}

Observe that by Equation (\ref{eq:powered_set_lemma2}) if $b$ is large, extensive amounts of information are needed, as $\#\overline{(A_j^{b})^*} $ will approximately increase at an order of at least $b^{n}$. We note that in most practical scenarios we find that we can keep the parameter $b$ small. Lemma \ref{lemma_simple} can be used to bound the number of expectations required when powering up moments  by enumerating the number of points `under' the powered-up corner points.
 
\begin{example}
Consider the regression model
$$
Y_3=\theta_{3,00}+\theta_{3,20}Y_1^2 +\theta_{3,01}Y_2 + v_3,$$
where $A_3=\{(0, 0), (2, 0), (0, 1) \}$, $A_3^*=\{(2, 0), (0, 1) \}$ and $\bar{A}_3^*=\{(0, 0), (1, 0), (2, 0), (0, 1) \}$. To calculate $Y_3^3$, by Lemma \ref{lemma_simple} we obtain the set of corner points $$(A_3^3)^*=\{(6,0), (4, 1), (2, 2), (0, 3) \},$$
and derive
$$
\overline{(A_3^3)}^*=\left\{(0, 0), (1, 0), (2, 0), (3, 0), (4, 0), (5, 0), (6, 0), (1, 1), (2, 1),  (3, 1), (4, 1), (2, 2), (0, 1), (0, 2), (0, 3) \right\}.
$$
\end{example}

\section{Message propagation}\label{sec:momentsuff}
In Section \ref{subsec:messagepass}
 we present an algorithm for generating and passing moments across $\mathcal{G}$, a key result of the paper. In Section \ref{subsec: explicitformula}  we consider the complexity, in terms of the number of donations, for the network of panels and discuss implications of various classes of models introduced in Section \ref{sec:graphical}.

\subsection{A general message propagation algorithm}
\label{subsec:messagepass}

\begin{lemma}\label{lemma:5.1} 
 
For any $Y_{i}$ and $Y_j$ in  $\mathcal{G}$, if $a_j(k)\geq 1, k=1, \dots, m$ in $\bb{a}_j=(a_j(1), a_j(2), \dots, a_j(m))$ and  $Y_{i}\notin \bb{Y}_{\Pi_j}$, then the parameters $\theta_{j\bb{a}_j}=0$ and their expectations  $\mathbb{E}[\theta_{j\bb{a}_j}]=0$.
\end{lemma}

\begin{proof}
For $\mathcal{G}$ to be valid, by definition $Y_j$ can only depend on its
non-dependents through its parents. Thus all
conditional expectations are functions of the values of the parent
variables alone.
\end{proof}

Note that the form of the exponent vector has changed as we consider the whole system depicted by $\mathcal{G}$, where previously individual panels were considered. This result states that the graph restricts the nature of the polynomial relationships that might exist, as you only need to consider monomials associated with parents. To evaluate the feasibility of the system model we need the size of the panel donations and the generations of the models of non-founder panels,  that is, panels that have a non-empty parent set.

\begin{assumption}\label{ass_indep}
      A set of variables represented by founder nodes in graph $\mathcal{G}$ are mutually independent.
\end{assumption} 
 
For the decision-maker $Q_{m+1}$, define the exponent set of utility function $A_{m+1}$. 
The algorithm for generating moments for each monomial in its utility has two phases: \emph{request and donate}.   
The \textbf{request phase} works backwards from the decision maker $Q_{m+1}$, iteratively generating polynomials  whose expectation will be requested by $Q_{j}$ from $Q_{i}$ for $i < j$. That is,  for $k=m+1, \dots, 3, 2$ determine $\Lambda_{k}^+$. Using these sets, the \textbf{donate phase} assembles the required expectations that needs to be produced by each panel to be passed forward, which is simply for $k=1, \dots, m$ to determine $\Lambda_k^-$. Denote by $\Lambda_{i\rightarrow j}$ the set of expectations of monomials provided by panel $Q_i$ to $Q_j$. Let $A_{i\rightarrow j}$ denote the set of exponent vectors corresponding to $\Lambda_{i\rightarrow j}$ that can be written as
\begin{equation}
    A_{i\rightarrow j}=\{\bb{a}: \bb{Y}^{\bb{a}}\in \Lambda_{i\rightarrow j}\},
\end{equation}
with $\bb{Y}=(Y_1, Y_2, \dots, Y_m)$ and $\bb{a}=(a(1), a(2), \dots, a(m))$.

\begin{definition}
\label{def:donate_request}
The set of all expectations of monomials \textbf{requested} by $Q_j$ from $\{Q_1, \dots, Q_{j-1} \}$ and the set of their associated exponent vectors are
\begin{equation}
    \Lambda_j^+=\bigcup_{i=1}^{j-1}\Lambda_{i\rightarrow j}, \quad A_j^+=\bigcup_{i=1}^{j-1}A_{i\rightarrow j}
\end{equation}
whereas the set of all expectations of monomials \textbf{donated} by $Q_j$ to $\{Q_{j+1}, \dots, Q_m, Q_{m+1} \}$ and the set of their associated exponent vectors are
\begin{equation}
    \Lambda_j^-=\bigcup_{i=j+1}^{m+1}\Lambda_{j\rightarrow i}, \quad  \quad A_j^-=\bigcup_{i=j+1}^{m+1}A_{j\rightarrow i}.
\end{equation}
\end{definition}

Note that in the request phase, the request is sent to the panel that corresponds to the \emph{leading term} in the interaction, the term that has the highest power on the product, as it is lexicographically ordered.

We note this is similar to a junction tree algorithm and propagation of clique tables associated with a Bayesian network \cite{Smith2010}. There are well-known message-passing algorithms that facilitate the computation of marginal posterior distributions \cite{Cowell2007, Dawid1992, Lauritzen1996, Spiegelhalter1993}. Here we focus on the development of an evidence propagation algorithm for Bayesian networks of variables modelled with polynomial regression by solely considering moments.

\begin{example}
\label{ex:special}
Consider panels $\{Q_1, Q_2, Q_3, Q_4\}$ with $\mathcal{G}$:
\[
\begin{array}
[c]{ccccc}%
 &  & Y_3 & \rightarrow & Y_4\\
 & \nearrow  &  & \nearrow & \\
Y_1 & \rightarrow & Y_2 &  &
\end{array}
\]
where
\begin{align*}
    Y_4&=\theta_{4, 0110}Y_2Y_3+v_4, \\
    Y_3&=\theta_{3,1000}Y_1+v_3, \\
    Y_2&=\theta_{2,1000}Y_1+v_2.
\end{align*}
Assume that the utility function only depends on $Y_4$. The backward algorithm requires the highest indexed parent of $Y_{4}$ to calculate the joint moments that include them, here $\mathbb{E}\left( Y_{2}Y_{3}\right)  $.
Then $Q_4$ requests from $Q_3$
$$
\mu_{0110}=\mathbb{E}\left(  Y_{2}Y_{3}\right)  =\mathbb{E}\left(
Y_{2}\left[  \theta_{3,1000}Y_{1}+v_{3}\right]  \right)  =m_{3,1000}\mathbb{E}%
\left(  Y_{1}Y_{2}\right)  =m_{3,1000}\mu_{1100},
$$
and $Q_3$ requests from $Q_2$
$$
\mu_{1100}=\mathbb{E}\left(  Y_{1}Y_{2}\right)  =\mathbb{E}\left(
Y_{1}\left[  \theta_{2,1000}Y_{1}+v_{2}\right]  \right)  =m_{2,1000}\mathbb{E}\left(
Y_{1}^{2}\right)  =m_{2,1000}\mathbb{\mu}_{2000},
$$
and $Q_2$ requests from $Q_1$ 
$$\mathbb{\mu
}_{2000}=\mathbb{E}\left(  Y_{1}^{2}\right)  = \mathbb{E}\left(
Y_{1}\right)  ^{2}+\mathbb{V}(Y_{1}).$$
The decision-maker can then calculate 
$$
\mu_{0001}=\mathbb{E}(Y_4)=\mathbb{E}(\theta_{4, 0110}Y_2Y_3+v_4)=m_{4,0110}m_{3,1000}m_{2,1000}\left(  \mathbb{E}\left(  Y_{1}\right)
^{2}+\mathbb{V}(Y_{1})\right).
$$
Then 
$$
\Lambda_{1}^{-}=\left\{  \mathbb{\mu}_{2000}\right\}  ,\Lambda
_{2}^{-}=\left\{  \mu_{1100}\right\}  ,\Lambda_{3}^{-}=
\left\{  \mu_{0110}\right\}  ,\Lambda_{4}^{-}=\left\{  \mu
_{0001}\right\},
$$
and all requests are 
$$
\Lambda_{2}^{+}=\left\{  \mathbb{\mu}_{2000}\right\}  ,\Lambda
 _{3}^{+}=\left\{  \mu_{1100}\right\}  ,\Lambda_{4}^{+}=
\left\{  \mu_{0110}\right\}  ,\Lambda_{5}^{+}=\left\{  \mu
_{0001}\right\}.
$$
These provide a full list of moments the decision maker will require to evaluate their expected utility.

\end{example}


 In Proposition \ref{th:polynomialder} the expection of monomial was considered for a single panel. The notation is modified to capture the expection of a monomial in the network. Define for a single donation from panel $Q_j$,  where $\boldsymbol{g}_j$ is a vector specifying the exponents of the monomial of interest,
\begin{align}
\label{eq:moment_def}
\mu_{\boldsymbol{g}_{j}}&=\mathbb{E}\left\{  Y_{j}^{g_{j}%
(j)}\left(
{\displaystyle\prod\limits_{i=1}^{j-1}}
 Y_{i}^{g_{j}(i)}\right) \right\},\\
 &=\mathbb{E}\left\{  \Big (\sum_{\bb{a}_j\in A_j}\theta_{j\bb{a}_j}\bb{Y}_{\Pi_j}^{\bb{a}_j}+ v_j \Big)^{g_{j}%
(j)}\left(
{\displaystyle\prod\limits_{i=1}^{j-1}}
 Y_{i}^{g_{j}(i)}\right) \right\}.
\end{align}

\begin{lemma}
\label{lemma:expect_monom}
\color{black}
It is sufficient for a panel
$Q_{j}$ to request from panel $Q_{i}$ $i=1, 2, \dots, j-1$ the non-central moments of $\bb{Y}$ with the corresponding set of exponent vectors $A_{i\rightarrow j}(\boldsymbol{g}%
_{j})$
\begin{align}
\label{eq:exponents_manipulation}
A_{i\rightarrow j}(\boldsymbol{g}_{j})=\left\{  \boldsymbol{a}%
_{i\rightarrow j}(\boldsymbol{g}_{j}%
):\mu_{\boldsymbol{a}_{i\rightarrow j}%
(\boldsymbol{g}_{j})}=\mathbb{E}\left(
{\displaystyle\prod\limits_{k=1}^{i}}
 Y_{k}^{a_{i\rightarrow j}(k)}\right)  \right\},
\end{align}
where $\boldsymbol{a}_{i\rightarrow%
j}(\boldsymbol{g}_{j})=\left(  a_{i\rightarrow
j}(1),a_{i\rightarrow j}(2),\ldots,a_{i\rightarrow j}(i),0,0,\ldots,0\right)
$ is given by
\begin{align}
\label{eq:power_pass2}
&\boldsymbol{a}_{i\rightarrow j}(\bb{g}_j)=\bb{b}_{j}+\boldsymbol{g}%
_{j}\left(  i\right), \quad \bb{b}_j\in A_j^{g_j(j)}
\end{align}
and $\boldsymbol{g}_{j} (i) = (  g_{j}(1),g_{j}%
(2),\dots,g_{j}(i),0,0,\dots,0)$.
\end{lemma}

\begin{proof}
\textcolor{black}{For sufficiency, substituting into Proposition \ref{th:polynomialder}, noting that under the message passing algorithm the panels below $Q_{j}$ defer the computation of any joint moments to the highest indexed such panel and relaxing the assumption that only parent set of $Y_j$, $\bb{Y}_{\Pi_j}$, is considered.  
}
\end{proof}
Define,
\begin{align}
\label{eq:A_rj}
A_{i\rightarrow j}=
{\displaystyle\bigcup\limits_{\boldsymbol{g}_{j}\in A_{j}^{-}}}
A_{i\rightarrow j}(\boldsymbol{g}_{j}),
\end{align}
\textcolor{black}{where term $i$ is the leading term in interactions. }
Then to derive Algorithm \ref{alg:cap} for generating and passing moments across $\mathcal{G}$, we combine the message passing algorithm in Section \ref{subsec:messagepass} with Lemma \ref{lemma:expect_monom}.

\begin{algorithm}
\caption{An algorithm for generating and passing moments across $\mathcal{G}$}\label{alg:cap}
\begin{algorithmic}
\STATE  For the decision-maker $Q_{m+1}$
derive $A_{m+1}^+$ and $\Lambda_{m+1}^+$

\FOR {$j=m, m-1, \dots, 2, 1$} 
  \STATE Generate sets $A_j^-=\bigcup_{k=j+1}^{m+1}A_{j\rightarrow k}$ and $\Lambda_{j}^-$ assembling the requests by panels $Q_{k}$, $k=j+1, \dots, m+1$
  \FOR {$\bb{g}_j\in A_j^-$}
  \STATE Obtain $A_{i\rightarrow j}(\bb{g}_j)$  $i=1, 2, \dots, j-1$ by Lemma \ref{lemma:expect_monom}
  \ENDFOR 
  \STATE Assemble exponent set $A_{i\rightarrow j} = \bigcup_{\bb{g}_j\in A_j^-}A_{i\rightarrow j}(\bb{g}_j)$ and $\Lambda_{i\rightarrow j}$, $i=1, 2, \dots, j-1$.
  \STATE Assemble set of exponent vectors $A_j^+=\bigcup_{i=1}^{j-1} A_{i\rightarrow j}$ and $\Lambda_j^+$
\ENDFOR
\STATE After all the requests have been sent, enter the donate phase and assemble the requested expectations that need to be produced by each panel to be passed forward to compute utility scores.
\end{algorithmic}
\end{algorithm}

\color{black}
\subsection{Complexity of a network of panels}
\label{subsec: explicitformula}

In Section \ref{subsec:complexity_panel} we considered the complexity of a single panel's calculations determined by the number of monomial terms that needed to be generated. When considering a network of panels, we proceed to characterise the complexity of the network by the number of donations that each panel needs to produce. Similarly, we can define sets of corner points $\big(A_{r\rightarrow j}\big)^*$ and $\big(A_r^- \big)^*$ and the corresponding closure sets $\overline{\big(A_{r\rightarrow j}\big)^*}$ and  $\overline{\big(A_r^- \big)^*}$ for $A_{r\rightarrow j}$ and $A_r^-$ respectively, for which we can write
\begin{align*}\label{eqn:subset_rewrite}
&\big(A_{r\rightarrow j}\big)^*\subseteq A_{r\rightarrow j}\subseteq \overline{\big(A_{r\rightarrow j}\big)^*},\\
&\big(A_r^- \big)^*\subseteq A_r^-\subseteq \overline{\big(A_r^- \big)^*}.
\end{align*}
We present results for the classes of polynomial regression models introduced in Section \ref{sec:graphical}.
\begin{corollary}
\label{thm:nonfounder}  
Suppose that all non-founder panels $Q_{j}$ use a full model  on a univariate output variable $Y_{j}$.  
Then $Q_i$, a panel that needs to donate to $Q_j$, will need to donate the expectations of monomials corresponding to
\begin{equation}
A_{i\rightarrow j}=\overline{\big(A_{i\rightarrow j}\big)^{\ast}}, \quad j=i+1,i+2, \dots, m+1,
\end{equation}
where the exponent set is the completion of corner points. The corner points in $\big(A_{i\rightarrow j}\big)^*$ are given by
\begin{align*}
\label{eq:corner_points_donations}
 \boldsymbol{a}_{i\rightarrow j}(\boldsymbol{g}_j)=\bb{b}_j+\boldsymbol{g}_{j}\left(  i\right) , \quad \text{where }\bb{b}_j\in \Big(A_{j}^{g_j(j)}\Big)^*\text{ and }\boldsymbol{g}_{j}\in\big(A_j^-\big)^*.
\end{align*}  
\end{corollary}
\begin{proof}
Start with a panel $Q_m$ and the exponent set $A_m$ for a regression model for $Y_m$. Since $Y_m$ is a full model, i.e. $A_m=\bar{A}_m^*$, we can deduct that $A_{i\rightarrow m}=\overline{\big(A_{i\rightarrow m} \big)^*}$, where $i\in \Pi_m$. For a panel $Q_{m-1}$, we can derive $A_{m-1}^-$ that contains exponents including $Y_{m-1}$ to a non-zero power. By Lemma \ref{lemma:expect_monom} and Equation (\ref{eq:exponents_manipulation}) we can deduce $A_{i\rightarrow m-1}(\bb{g}_{m-1})$, where $\bb{g}_{m-1}\in A_{m-1}^-$. Since $A_{m-1\rightarrow m}=\overline{\big(A_{m-1\rightarrow m}\big)^*}$ and $Y_{m-1}$ is a full model, we can derive that $A_{i\rightarrow m-1} = \overline{\big(A_{i\rightarrow m-1}\big)^*}$. Therefore, proceeding iteratively, we can demonstrate that $A_{i\rightarrow j}=\overline{\big(A_{i\rightarrow j}\big)^*}$ holds in general for $j=i+1, \dots, m-1$. To derive the corner points in $\big(A_{i\rightarrow j}\big)^*$, we use Lemma \ref{lemma:corner_points1} with $b=1$ and Lemma \ref{lemma_simple}.

\end{proof}
This result implies that the necessary message-passing between panels can be computationally expensive, but still manageable. We introduce a simple example to demonstrate this.
\begin{example}
   Consider $\mathcal{G}$ in Example \ref{ex:special}  with a full regression model for all non-founder variables:
    \begin{align*}
Y_{4}  & =\theta_{4,0000}+\theta_{4,0100}Y_{2}+\theta_{4,0010}Y_{3}+\theta_{4,0200}%
Y_{2}^{2}+\theta_{4,0110}Y_{2}Y_{3}+\theta_{4,0020}Y_{3}^{2}+v_{4},\\
Y_{3}  & =\theta_{3,0000}+\theta_{3,1000}Y_{1}+\theta_{3,2000}Y_{1}^{2}+v_{3},\\
Y_{2}  & =\theta_{2,0000}+\theta_{2,1000}Y_{1}+v_{2}.
\end{align*}
Assuming the utility function only depends on $Y_4$, 
we have that
$A_4 = \bar{A}_4^*$, where $$A_4^*=\{ (0, 2, 0, 0), (0, 1, 1, 0), (0, 0, 2, 0)\}$$
and $A_4^-=\{(0, 0, 0, 1) \}$.
By Corollary \ref{thm:nonfounder}, we can therefore obtain that:
\begin{align*}
    &\big(A_{3\rightarrow 4}\big)^* = \{(0, 0, 2, 0), (0, 1, 1, 0) \}, \quad \text{and }A_{3\rightarrow 4} = \overline{\big(A_{3\rightarrow 4}\big)^*}=\{(0, 0, 1, 0), (0, 0, 2, 0), (0, 1, 1, 0) \}, \\
    &\big(A_{2\rightarrow 4}\big)^* = \{(0, 2, 0, 0) \}, \quad \text{and }A_{2\rightarrow 4} = \overline{\big(A_{2\rightarrow 4}\big)^*}=\{(0, 1, 0, 0), (0, 2, 0, 0) \}.
\end{align*}
We then have $A_3^-=\{(0, 0, 1, 0), (0, 0, 2, 0), (0, 1, 1, 0) \}$, and the set of corner points of the regression model for $Y_3$ is $A_3^*=\{(2,0,0, 0) \}$. By Corollary \ref{thm:nonfounder}, we can derive that 
\begin{align*}
    &\big(A_{2\rightarrow 3}\big)^* = \{(2,1,0, 0) \}, \quad \text{and }A_{2\rightarrow 3} = \overline{\big(A_{2\rightarrow 3}\big)^*}=\{(0, 1, 0, 0), (1, 1, 0, 0), (2, 1, 0, 0) \}, \\
    &\big(A_{1\rightarrow 3}\big)^* = \{(4, 0, 0, 0) \}, \quad \text{and }A_{1\rightarrow 3} = \overline{\big(A_{1\rightarrow 3}\big)^*}=\{(1,0,0, 0), (2,0,0, 0), (3,0,0, 0), (4,0,0,0) \}.
\end{align*}
For panel $Q_2$, we find $A_2^- = \{(0,1,0, 0), (1,1,0, 0), (2,1,0, 0), (0,2,0, 0) \}$ and the set of corner points 
is $A_2^*=\{(1,0,0, 0) \}$, and
\begin{align*}
    &\big(A_{1\rightarrow 2}\big)^* = \{(3,0,0, 0) \}, \quad \text{and }A_{1\rightarrow 2} = \overline{\big(A_{1\rightarrow 2}\big)^*}=\{(1,0,0,0), (2,0,0, 0), (3,0,0, 0) \}.
\end{align*}
\end{example}
It follows that we can identify a class of graphs, P\'olya forests \cite{Pearl1988}, used with a utility function with attributes only on the terminal nodes of $\mathcal{G}$. A simple case also occurs when $\mathcal{G}$ is a directed tree and the utility function is linear in its different components. In a Bayesian decision analysis this translates to a situation when the components of the utility are preference-independent attributes whose marginal utilities are all polynomial \cite{Leonelli2020}.

\begin{corollary}
\label{cor:tree}
For $d\in \mathcal{D}$, let   $U(\bb{Y},d)$ have the form
\begin{align*}
U(\bb{Y},d)=\sum_{i=1}^{m}k_iU_{i}(Y_i, d)=\sum_{i=1}^m k_i\Big(\sum_{k=0}^{c_{i}}\rho_{ik}(d)Y_{i}^{k}\Big),
\end{align*}
where the coefficients $\rho_{ik}(d)\in\mathbb{R}$. Suppose each $Y_{i}$, $i=2,\ldots,m$ has at most one
parent, i.e., $\bb{Y}_{\Pi_i}=\{X_i \}$, so that $\mathcal{G}$ is a forest. Then all non-founder variables are governed by a simple model defined as
\begin{align*}
Y_i=\sum_{j=0}^{b_{i}}\theta_{ij}X_{i}^{j} +v_i.
\end{align*}
For the decision-maker $Q_{m+1}$ to obtain the expected utility, panel $Q_{i}$ is required 
to donate $\Lambda_{i}^{-}$ with the corresponding set of exponent vectors $A_i^-=\overline{(A_i^-)^*}$ with
$(A_{i}^-)^{\ast}= \{   (  0,\ldots0,a_i(i),0,\ldots,0 )   \}$
where $a_{i}(i)$ is defined 
\begin{equation}
\label{eq:recursive_cal_mod}
  a_{i}(i)=\begin{cases}
    c_i  &\text{if }Y_i\text{ is a terminal node}, \\
\max\Big(c_i,  \max_{k: (i, k)\in E(\mathcal{G})}\left\{
b_{k}a_{k}(k)\right\}\Big)&\text{otherwise},
  \end{cases}
\end{equation}
and $E(\mathcal{G})$ is an edge set of ordered pairs of vertices in graph $\mathcal{G}$.
\color{black}
\end{corollary}

\begin{proof}
These results can be obtained directly from Corollary \ref{thm:nonfounder} by substituting the conditions above.
\end{proof}
Thus each $Q_i$ needs to provide expectations of $Y_i$ up to the power $a_i(i)$ given above. We consider how the depth of the tree, the complexity of the utility function, and regression models for individual variables can significantly increase the number of non-central moments that $Q_i$ needs to deliver.

\begin{example}
    Consider panels $\{Q_1, Q_2, Q_3, Q_4, Q_5, Q_6 \}$ and decision center $Q_7$ with utility function of the form:
    \begin{align*}
        U(\bb{Y}, d) &= k_1\Big(\sum_{k=0}^7 \rho_{1k}(d)Y_1^k \Big) + k_2\Big(\sum_{k=0}^1 \rho_{2k}(d)Y_2^k \Big) + k_3\Big(\sum_{k=0}^4 \rho_{3k}(d)Y_3^k \Big)+ k_4\Big(\sum_{k=0}^5 \rho_{4k}(d)Y_4^k \Big) \\
        &+k_5\Big(\sum_{k=0}^2 \rho_{5k}(d)Y_5^k \Big) + k_6\Big(\sum_{k=0}^4 \rho_{6k}(d)Y_6^k \Big)
    \end{align*}
       and with $\mathcal{G}$:
\[
\begin{array}
[c]{ccccccc}
& & & & & & Y_4\\
& & & & & \nearrow  &\\
Y_{1} & \rightarrow  & Y_{2} & \rightarrow & Y_{3} & \rightarrow & Y_{5}\\
& & & & & \searrow & \\
& & & & & & Y_{6}
\end{array}
\]
For each non-founder vertex $Y_j, j=1, \dots, 6$, we specify regression models:
\begin{align*}
    Y_6&=\sum_{k=0}^3\theta_{6,k}Y_3^k+v_6, \quad Y_5=\sum_{k=0}^1\theta_{5,k}Y_3^k+v_5,\quad Y_4=\sum_{k=0}^7\theta_{4,k}Y_3^k+v_4,\\
    Y_3&=\sum_{k=0}^2\theta_{3,k}Y_2^k+v_3,\quad
    Y_2  =\sum_{k=0}^4\theta_{2,k}Y_1^k+v_2.
\end{align*}
The donations from the terminal nodes are derived directly from the utility function:
\begin{align*}
    &G_4=A_4^-=\overline{(A_4^-)^*}\text{ with }(A_4^-)^*=\big\{(0,0,0,5,0,0) \big\}, \\
    &G_5=A_5^-=\overline{(A_5^-)^*}\text{ with }(A_5^-)^*=\big\{(0,0,0,0,2,0) \big\},\\
    &G_6=A_6^-=\overline{(A_6^-)^*}\text{ with }(A_6^-)^*=\big\{(0,0,0,0,0,4) \big\}.
\end{align*}
We next obtain those donations required from the remaining panels, where we are effectively comparing the donations to children nodes as well as to the decision center using Equation (\ref{eq:recursive_cal_mod}):
$$
G_3 = A_3^-=\overline{(A_3^-)^*}\text{ with }(A_3^-)^*=\{(0,0,35,0,0,0) \},  
$$
since
$$
a_3(3) = \max\big(c_3, \max(b_4a_4(4), b_5a_5(5), b_6a_6(6)) \big) = \max\big(4, \max(35, 2, 12) \big) = 35.
$$
Similarly, we can calculate that 
\begin{align*}
&G_2=A_2^-=\overline{(A_2^-)^*}\text{ with }(A_2^-)^*=\big\{(0,70,0,0,0,0) \big\}, \\
&G_1=A_1^-=\overline{(A_1^-)^*}\text{ with }(A_1^-)^*=\big\{(280,0,0,0,0,0) \big\}.
\end{align*}
\end{example}
Note here that the most detailed information, in terms of the number of moments, is required from the panels lowest in the tree.
We conclude this section with the worst-case scenario for graphical models where the underlying graph is complete. We present an explicit bound for the number of moments that need to be generated to compute the expected utility when each model is a graphical regression model.

\begin{corollary}
\label{cor:number_of_moments}
Suppose that both the utility function and all the models are graphical on
the maximum number of regressors on a complete graph $\mathcal{G}$. Then the total
number of non-central moments, denoted as $\#\left(  A^{-}\right)$, that all panels will
need to donate for the expected utility scores to be calculated is 

$$
 \#\left(  A^{-}\right)=\sum_{j=1}^{m} \# \overline{\big(A_j^-\big)^*}=\sum_{j=1}^m \Big(a^*_j(j)+1 \Big)^{j-1}\times a^*_j(j),
$$
where $a^*_j(j)$ is the $j$\textsuperscript{th} element of the corner point $\bb{a}^*_j$ and $
\big(A_j^-\big)^*=\{(a^*_j(1), \dots, a^*_j(j), 0, \dots, 0) \}.$

\end{corollary}

\begin{proof}
Rewriting the total number of non-central moments that all panelists will need to donate as 
\begin{equation*}
    \#\left(  A^{-}\right)=\sum_{j=1}^m\#\Big(A_j^- \Big) =\sum_{j=1}^m\#\Big(\bigcup_{i=j+1}^{m+1} A_{j\rightarrow i}\Big). 
\end{equation*}
Since we consider the graphical regression model (a class of full models) for the utility function and all the variables in $\mathcal{G}$, by Corollary \ref{thm:nonfounder}, we can write down 
\begin{equation*}
    \#\left(  A^{-}\right)=\sum_{j=1}^m\#\Big(\bigcup_{i=j+1}^{m+1} A_{j\rightarrow i}\Big)=\sum_{j=1}^m\#\Big(\bigcup_{i=j+1}^{m+1} \overline{\big(A_{j\rightarrow i}\big)^*}\Big).
\end{equation*}
Since we have a complete graph $\mathcal{G}$, the expression can be further simplified to
\begin{equation*}
    \#\left(  A^{-}\right) = \sum_{j=1}^{m} \# \Big(\overline{\big(A_{j\rightarrow j+1}\big)^*}\Big).
\end{equation*}
By definition of the graphical regression model, we write down $(A_{m\rightarrow m+1} )^*=\{(1, 1, \dots, 1) \}$
and iteratively applying the results from Corollary \ref{thm:nonfounder}, we can derive that the corner points in each exponent set, $A_{m-r\rightarrow m-r+1}$ for $r=1, \dots, m-1$,  has the form 
\begin{equation*}
\bb{a}_{m-r\rightarrow m-r+1}(\bb{g}_{m-r+1})=\big(a_{m-r\rightarrow m-r+1}(1),a_{m-r\rightarrow m-r+1}(2) \dots,a_{m-r\rightarrow m-r+1}(m)),
\end{equation*}
where $\bb{g}_{m-r+1}\in \big(A_{m-r+1}^- \big)^*$ and 
$$
a_{m-r\rightarrow m-r+1}(k)=\begin{cases}
    2^r &\text{ if } k=1, \dots, m-r\\
    0 &\text{ if } k=m-r+1, \dots, m.
\end{cases}
$$
We recognise that instead of computing the set of exponents corresponding to donations of panel $Q_i$ to each panel $Q_j$, $j=i+1, \dots, m, m+1$, we need to work backwards starting with the corner point for utility function and deriving the corresponding exponents from each panel requested and then obtaining the closure of the corner points set. We can write down that
$$
\big(A_j^-\big)^*=\big(A_{j\rightarrow j+1}\big)^*=\{(a^*_j(1), \dots, a^*_j(j), 0, \dots, 0) \} \text{ and } \overline{\big(A_j^-\big)^*}=\overline{\big(A_{j\rightarrow j+1}\big)^*}.
$$
Aggregating over those moments dominated by these corner points, we can conclude that 
\begin{equation}
\#\left(  \overline{\big(A_j^-\big)^*}\right)= \Big(a^*_j(j)+1 \Big)^{j-1}\times a^*_{j}(j),
\end{equation}
since the $j$\textsuperscript{th} component of the exponent vector in $\overline{\left(A_j^-\right)^*}$ can take any value in $\{1, \dots,   a^*_j(j)\}$, while $j+1, \dots, m$ entries in this exponents vector are fixed at zero. At the same time, the first $j-1$ entries contain all possible combinations of $\{0, 1, \dots,  a^*_j(j)\}$ represented by $\left(a^*_j(j)+1 \right)^{j-1}$.
\end{proof}

By Corollary \ref{cor:number_of_moments}, it can be checked that the total number of moments that need to be considered becomes large when $m \geq 4$, i.e. when $m=4$, $\# (A^-)=54$ and for $m=5$, $\#(A^-)=258$. This illustrates the importance of building conditional independence assumption into the structure of the graph.

\section{Application to two case studies}
\label{sec:application} 
To demonstrate the methodology we present two cases studies. In Section \ref{subsec:energy_plan}, in an energy planning domain, we study a situation where there are multiple expert panels and where data is hard to obtain. In Section \ref{subsec:blood_gas}, we consider the decision problem centered around the respiratory condition of a patient in an Intensive Care Unit (ICU) in a hospital, where there is essentially one expert panel, a team of medical specialists, with full access to informing data. We show that in this setting the calculations for the methodology can still be supported.

 \subsection{An application in energy planning} 
\label{subsec:energy_plan}
This study is a simplified version of a decision support tool that has been developed with a UK county council to assess the impacts of various heating technologies on operational costs and CO\textsubscript{2}-equivalent emissions. This complex process is characterized by multiple interdependent processes and relies on various sources of information. We demonstrate the type of graphical model for decision support in which various panels need to be coordinated to provide information for assessing different decision options.

We identified variables to be considered: surface temperature $(Y_1)$, fuel prices $(Y_2)$, technologies $(Y_3)$, heating demand $(Y_4)$, capital investment $(Y_5)$, energy usage $(Y_6)$, fuel costs $(Y_7)$, carbon emissions $(Y_8)$ and investment costs $(Y_9)$; with the conditional independence structure illustrated in Figure \ref{fig:EnergyDAG}. 
For simplicity, we omit details of the form of attributes, as this would distract from the discussion below.

 \begin{figure}[h!]
\begin{center}
 \resizebox{0.7\textwidth}{!}{%
\begin{tikzpicture}  
    \node[] (a) at (0,0) {Energy Usage $(Y_6)$};
    \node[] (b) [right =of a] {Carbon Emissions $(Y_8)$};  
    \node[] (d) [below =of a] { Capt. Invest $(Y_5)$};  
    \node[] (c) [below =of b] {Invest. Cost $(Y_9)$};
    \node[] (e) [above =of b] {Fuel costs $(Y_7)$};
    \node[] (f) [left =of a] {Heating Demand $(Y_4)$};
    \node[] (g) [below =of f] {Technologies $(Y_3)$};
    \node[] (h) [above =of f] {Surface Temp. $(Y_1)$};
    \node[] (l) [above =of a] {Fuel prices ($Y_2$)};
    \path (a) edge (b); 
    \path (d) edge (c); 
    \path (d) edge (a);
    \path (a) edge (e);
    \path (f) edge (a);
    \path (g) edge (a);
    \path (h) edge (f);
    \path (l) edge (a);
    \path (g) edge (d);
    \path (l) edge (e);
\end{tikzpicture} 
 }%
\end{center}
     \caption{Graphical model for an energy planning application. } 
    \label{fig:EnergyDAG}
\end{figure}
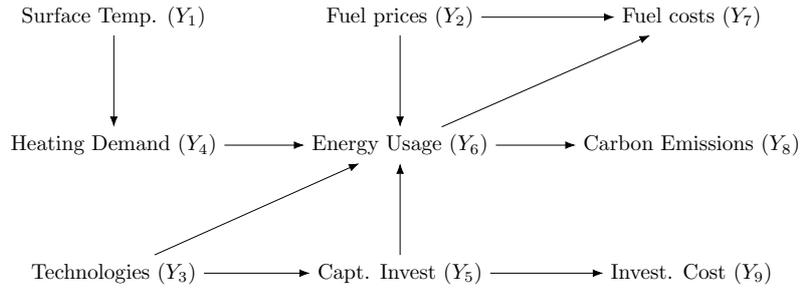
Assume that a decision center, $Q_{10}$, has a utility function linear in the terminal nodes (fuel costs, carbon emissions and investment costs) and   assume that decisions are not arguments of the utility function,
\begin{equation}
    U(\bb{Y}, d)=k_7Y_7+k_8Y_8+k_9Y_9,
\end{equation}
and assume the decision center can compute the expected utility for decision options,
\begin{equation}
\overline{U}(d) = k_7\E{Y_7}+k_8\E{Y_8}+k_9\E{Y_9}.
\end{equation}
 Panel models are specified, for example, fuel costs are directly affected by fuel prices (panel $Q_2$) and energy usage (panel $Q_6$)  \cite{Barons2021,Hall2016},
\begin{align*}
    Y_7&= \theta_{7,000000} + \theta_{7,010000}Y_2+\theta_{7,000001}Y_6+\theta_{7,010001}Y_2Y_6+v_7, 
    \end{align*}
     while carbon emissions depend on energy usage,
    \begin{align*}
    Y_8&=\theta_{8,0000000}+\theta_{8,0000010}Y_6+v_8, 
        \end{align*}
        and investment costs are primarily driven by capital investments,  
    \begin{align*}
    Y_9&=\theta_{9,00000000} +\theta_{9,00001000}Y_5+v_9.
\end{align*}
Further we have,
\begin{align*}
Y_6 &= \theta_{6,00000} + \theta_{6,01000}Y_2+\theta_{6,00010}Y_4+\theta_{6,00100}Y_3+\theta_{6,00001}Y_5+\theta_{6,00101}Y_3Y_5+v_6,\\
Y_5&=\theta_{5,0000}+\theta_{5,0010}Y_3+v_5,\\
Y_4&=\theta_{4,000}+\theta_{4,100}Y_1+v_4.
\end{align*}
We omit the specification of statistical models for founder nodes, since we are interested in demonstrating how the set of all expectations of monomials can be determined by applying Algorithm \ref{alg:cap}. Start by considering the following set of expectation of monomials and the set of their associated exponent vectors:
\begin{align*}
&\Lambda_{10}^+ = \{ \E{Y_7}, \E{Y_8}, \E{Y_9}\}, \\ &A_{10}^+=\{(0, 0, 0, 0, 0, 0, 1, 0, 0), (0, 0, 0, 0, 0, 0, 0, 1, 0), (0, 0, 0, 0, 0, 0, 0, 0, 1) \}.
\end{align*}
In other words, $Q_{10}$ sends the requests to $G_7$, $G_8$ and $G_9$ to provide the decision centre with $\E{Y_7}$, $\E{Y_8}$ and $\E{Y_9}$ respectively. Based on these requests from $Q_{10}$ and topology of the graph, we can derive
\begin{align*}
    &\Lambda_9^-=\{\E{Y_9} \}, \quad A_9^-=\{(0, 0, 0, 0, 0, 0, 0, 0, 1) \},\\
    &\Lambda_8^-=\{\E{Y_8} \}, \quad A_8^-=\{(0, 0, 0, 0, 0, 0, 0, 1, 0) \},\\
    &\Lambda_7^-=\{\E{Y_7} \}, \quad A_7^-=\{(0, 0, 0, 0, 0, 0, 1, 0, 0) \}.
\end{align*}
To donate these moments to the decision centre, panels $Q_9$, $Q_8$ and $Q_7$ in turn need to request
\begin{align*}
    &\Lambda_{5\rightarrow 9}=\{\E{Y_5} \}, \quad A_{5\rightarrow 9}=\{(0, 0, 0, 0, 1, 0, 0, 0, 0) \},\\
    &\Lambda_{6\rightarrow 8}= \{\mathbb{E}[Y_6] \}, \quad A_{6\rightarrow 8}=\{(0, 0, 0, 0, 0, 1, 0, 0, 0) \},\\
    &  \Lambda_{2\rightarrow 7} = \{\E{Y_2} \}, \quad A_{2\rightarrow 7} = \{(0, 1, 0, 0, 0, 0, 0, 0, 0) \},\\
    &\Lambda_{6\rightarrow 7} = \{\E{Y_6}, \E{Y_2Y_6} \}, \quad A_{6\rightarrow 7} =\{(0, 0, 0, 0, 0, 1, 0, 0, 0), (0, 1, 0, 0, 0, 1, 0, 0, 0) \}.
\end{align*}
Therefore we can derive the set of expectations of monomials requested by each panel
$$
    \Lambda_9^+=\{ \E{Y_5}\},
    \Lambda_8^+=\{\E{Y_6} \}, \Lambda_7^+=\{\E{Y_2}, \E{Y_6}, \E{Y_2Y_6}\}.
$$
We then consider the set of expectations of monomials donated by $Q_6$ and the set of their associated exponent vectors:
$$
\Lambda_{6}^- = \{\E{Y_6}, \E{Y_2Y_6} \}, \quad A_{6}^- =\{(0, 0, 0, 0, 0, 1, 0, 0, 0), (0, 1, 0, 0, 0, 1, 0, 0, 0) \}.
$$
Then the moments requested by panel $Q_6$ to perform the donations above are 
\begin{align*}
   & \Lambda_{2\rightarrow 6}=\{\mathbb{E}[Y_2], \mathbb{E}[Y_2^2] \}, \
    \Lambda_{3\rightarrow 6}=\{\mathbb{E}[Y_3], \mathbb{E}[Y_2Y_3] \},\
     \Lambda_{4\rightarrow 6}=\{\mathbb{E}[Y_4], \mathbb{E}[Y_2Y_4] \}, \\
 &   \Lambda_{5\rightarrow 6}=\{\mathbb{E}[Y_5], \mathbb{E}[Y_2Y_5], \mathbb{E}[Y_3Y_5], \mathbb{E}[Y_2Y_3Y_5] \}.
\end{align*}
that can be combined to produce the set of all expectations requested by the panel $Q_6$:
$$
\Lambda_6^+=\{\E{Y_2}, \E{Y_2^2}, \E{Y_3}, \E{Y_2Y_3}, \E{Y_4}, \E{Y_2Y_4}, \mathbb{E}[Y_5], \mathbb{E}[Y_2Y_5], \mathbb{E}[Y_3Y_5], \mathbb{E}[Y_2Y_3Y_5]\}.
$$
We proceed to panel $Q_5$ and derive the set of expectations that the panel needs to deliver:
\begin{align*}
    \Lambda_{5}^-=\{\mathbb{E}[Y_5], \mathbb{E}[Y_2Y_5], \mathbb{E}[Y_3Y_5], \mathbb{E}[Y_2Y_3Y_5] \}.
\end{align*}
Then the moments requested by panel $Q_5$,
\begin{align*}
    \Lambda_{2\rightarrow 5}=\{\mathbb{E}[Y_2] \}, \
    \Lambda_{3\rightarrow 5}=\{\mathbb{E}[Y_3], \mathbb{E}[Y_3^2], \mathbb{E}[Y_2Y_3], \mathbb{E}[Y_2Y_3^2] \},
\end{align*}
with $\Lambda_5^+=\{\mathbb{E}[Y_2], \mathbb{E}[Y_3], \mathbb{E}[Y_3^2], \mathbb{E}[Y_2Y_3], \mathbb{E}[Y_2Y_3^2]\}$.
We consider the set of expectations that panel $Q_4$ needs to donate $\Lambda_4^-=\{\E{Y_4}, \E{Y_2Y_4} \}$ and in turn the moments requested by panel $Q_4$,
\begin{align*}
    \Lambda_{1\rightarrow 4}=\{\mathbb{E}[Y_1] \}, \
    \Lambda_{2\rightarrow 4}=\{\mathbb{E}[Y_2], \mathbb{E}[Y_1Y_2] \}.
\end{align*}
As previously, the sets can be combined together to obtain the set of all expectations of monomials requested by panel $Q_4$, i.e. $\Lambda_4^+=\{\E{Y_1}, \E{Y_2}, \E{Y_1Y_2} \}$.

Within this graphical model the founder nodes $Y_1$, $Y_2$ and $Y_3$ have been assumed to be marginally independent (by Assumption \ref{ass_indep}). We can therefore specify 
$$
\Lambda_3^-=\{\E{Y_3}, \E{Y_3^2}, \E{Y_2Y_3}, \E{Y_2Y_3^2} \},
$$
which means that $Q_3$ needs to request $\Lambda_{2\rightarrow 3} = \{\mathbb{E}[Y_2] \}$. Similarly, we can derive that $\Lambda_{1\rightarrow 2} = \{\mathbb{E}[Y_1] \}.$ After all the requests have been sent, we enter the donate phase and assemble the requested expectations that need to be produced by each panel to be passed forward to compute utility scores. 
This is a primary example for decision making where there are many panels, and it highlights advantages of the modularity imposed by the graphical model and the form of utility function.

 
 \subsection{Clinical decision support system}
\label{subsec:blood_gas}
 
Consider a decision problem centered around improving a patient's respiratory condition in the ICU. We provide a demonstration of a clinical decision support system (CDSS) based on the proposed algorithm for propagating moments. This CDSS incorporates expert knowledge in the algorithm, combining the capabilities of two classes of CDSSs, i.e., knowledge-based and non-knowledge based \cite{Sutton2020}. In particular, the CDSS is composed of the algorithm used to model the decision based on the relationships between clinical variables of interest together with the utility function elicited from medical experts. This inference engine can then be applied to the patient's clinical data to generate an output or action.

For a patient, we investigate the effectiveness of two decisions $\mathcal{D}=\{d_{0}, d_{1} \}$ on their respiratory condition; keeping a patient on the low ventilator settings ($d_{0}$) or increasing the ventilator settings ($d_{1}$). During normal breathing, if CO\textsubscript{2} production increases, the respiratory minute volume, the volume of gas inhaled or exhaled from a person's lung per minute, is increased to reduce partial CO\textsubscript{2}. In a fully mechanically ventilated patient, minute volume is completely controlled by clinicians, and therefore used to regulate pH and partial CO\textsubscript{2} levels by changing pressure/volume/rates in ventilator settings. Nevertheless, it's crucial to recognise that for patients with severely compromised lung function, increased pressure, volume, or rates in ventilator settings might worsen lung damage. In addition, oxygen delivery is crucial in managing a patient's respiratory condition. Clinicians typically rely on SpO\textsubscript{2} measurements, haemoglobin oxygen saturations, to decide on oxygen therapy needs. However, SpO\textsubscript{2} levels alone may not accurately reflect oxygen delivery due to variations caused by changes in local and systemic chemical characteristics, such as the pH or temperature surrounding the red cell containing haemoglobin. This means that high SpO\textsubscript{2} levels do not guarantee optimal oxygen delivery to cells, and relying solely on SpO\textsubscript{2} for treatment decisions can lead to over or under-treatment, potentially causing harm. Therefore we include the p50 value, the oxygen tension at which haemoglobin is saturated with oxygen at 50\%, as a summary measure of haemoglobin-oxygen affinity. Haemoglobin-oxygen affinity describes the relationship between haemoglobin oxygen saturation (SpO\textsubscript{2}) and the partial pressure of oxygen (PO\textsubscript{2}) and is represented by oxygen dissociation curve (ODC). The ODC is not static, in particular, changes in pH, temperature, 2,3-DPG and partial pressure of carbon dioxide in the blood (PCO\textsubscript{2}) can cause a shift in ODC, and during these shifts the most marked alternations occur in the middle part of the curve around p50 value \cite{Morgan1999}.

Based on this information as part of utility function specification, we consider abnormal deviations from the reference values in three clinical measurements that are mostly impacted by poor respiratory condition: partial pressure of carbon dioxide, pH and hemoglobin-oxygen affinity measure. The physiological variables have been modified to depict the abnormal changes, in particular $Y_{1t}=\ln\Big(\frac{\text{pCO\textsubscript{2}}_t}{5.33} \Big)$, $Y_{2t} = (\text{T}_t-37)$, $Y_{3t} =(\text{pH}_t-7.40)$, $Y_{4t}=\text{cDPG}_t-5$ and  $Y_{5t} = (p50_t-3.47)$, based on the specifications from blood gas machine \cite{Siggaard1974}. The conditional independence structure is depicted in Figure \ref{fig:ODC}, where it can be seen  
that changes in pH  ($Y_{3t}$), temperature T ($Y_{2t}$), partial CO\textsubscript{2} ($Y_{1t}$) and cDPG, which is a proxy to 2.3-DPG concentrations, ($Y_{4t}$) have an effect on p50 value ($Y_{5t}$), haemoglobin-oxygen affinity measure \cite{Siggaard1974}.

Since we observe all data simultaneously, we assume all the data available up to $t-1$, generating possible decision options for time $t$. An appropriate utility function could take the form:

\begin{equation}
    U(\bb{Y}_t, d) = -k_1Y_{1t} + k_3Y_{3t}-k_5Y_{5t} - k_d\rho(d),
\end{equation}
where the criterion weights $k_d$ and $k_i\in (0,1)$, and $\rho(d)$ is a proxy measure for lung damages associated with decision $d$. 

\begin{figure}[h!]
\begin{center}
 \resizebox{0.4\textwidth}{!}{%
    \begin{tikzpicture}[roundnode/.style={circle, very thick, minimum size=7mm,  draw=black},
    ]  
    \node[] (a) at (1,0) {pCO\textsubscript{2} ($Y_{1t}$) };
    \node[] (c) at (2, -2){ pH ($Y_{3t}$)};
    \node[] (d) at (2, -4) {cDPG ($Y_{4t}$)};
    \node[] (f) at (6, -2) {p50 ($Y_{5t}$)};
    \node[] (g) at (1, -6) {T ($Y_{2t}$)};
    \draw[->] (a) edge (f);
    \draw[->] (a) edge (c);
    \draw[->] (d) edge (f);
    \draw[->] (c) edge (d);
    \draw[->] (c) edge (f);
    \draw[->] (g) edge (d);
    \draw[->] (g) edge (f);
\end{tikzpicture} 
}
\end{center}
\caption{Graphical model representing changes in haemoglobin-oxygen affinity.}
\label{fig:ODC}
\end{figure}
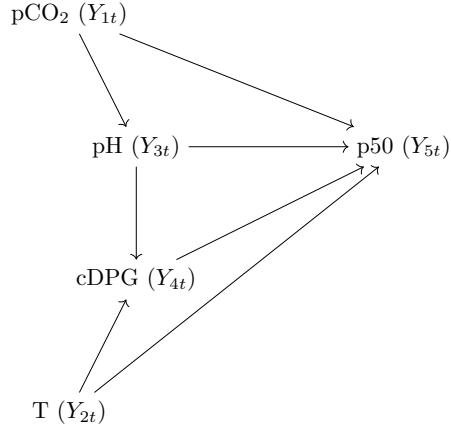

All variables are measured every 4-6 hours from the patient's arrival in the ICU to recovery/death. We therefore consider temporal models for multivariate time series  using the Multidimensional regression model (MDM) \cite{Queen1993}. 
Under the MDM, the joint distribution breaks down into separate components, each of which consists of conditional distributions that are univariate dynamic linear regression models (DLMs) \cite{West1997}. We specify observation equations for the variables in Figure \ref{fig:ODC}:
\begin{align*}
     Y_{5t}&=\theta_{5t,0000} +\theta_{5t,1000}Y_{1t}+\theta_{5t,0100}Y_{2t} +\theta_{5t,0010}Y_{3t}+\theta_{5t,0001}Y_{4t}\\
&+\theta_{5t,1010}Y_{1t}Y_{3t}+\theta_{5t,0101}Y_{2t}Y_{4t}+\theta_{5t,0011}Y_{3t}Y_{4t} + v_{5t},
\end{align*}
and for the remaining non-founder nodes we have:
\begin{align*}
    Y_{4t} = \theta_{4t,000} + \theta_{4t,010}Y_{2t}+\theta_{4t,001}Y_{3t} + v_{4t},\quad
    Y_{3t}=\theta_{3t,0} + \theta_{3t,1}Y_{1t} + v_{3t}.
\end{align*}
We also specify a simple steady state DLM \cite{West2006} for the founder nodes,  
\begin{align*}
    Y_{2t}=\theta_{2t,0} + v_{2t}, \quad
    Y_{1t}= \theta_{1t,0} + v_{1t}.
\end{align*}
As part of DLM specification, we need to define the temporal evolution of the system equations,
\begin{equation*}
    \bb{\theta}_{it} = G_{it}\bb{\theta}_{it-1} + \bb{\omega}_{it},
\end{equation*}
with $v_{i,t}\sim N(0, V_{it})$ and $\bb{\omega}_{it}\sim N(\bb{0}_d, W_{it})$. We assume that the errors in observation and system equations are independent, and $G_{it}$ is an identity matrix. To fit the DLMs, we use the \texttt{R} library \texttt{dlm} available from \texttt{CRAN} \cite{Petris2010}.

\label{subsec:blood_gas}
\begin{figure}[h!]
\begin{center}
\includegraphics[scale=0.6]{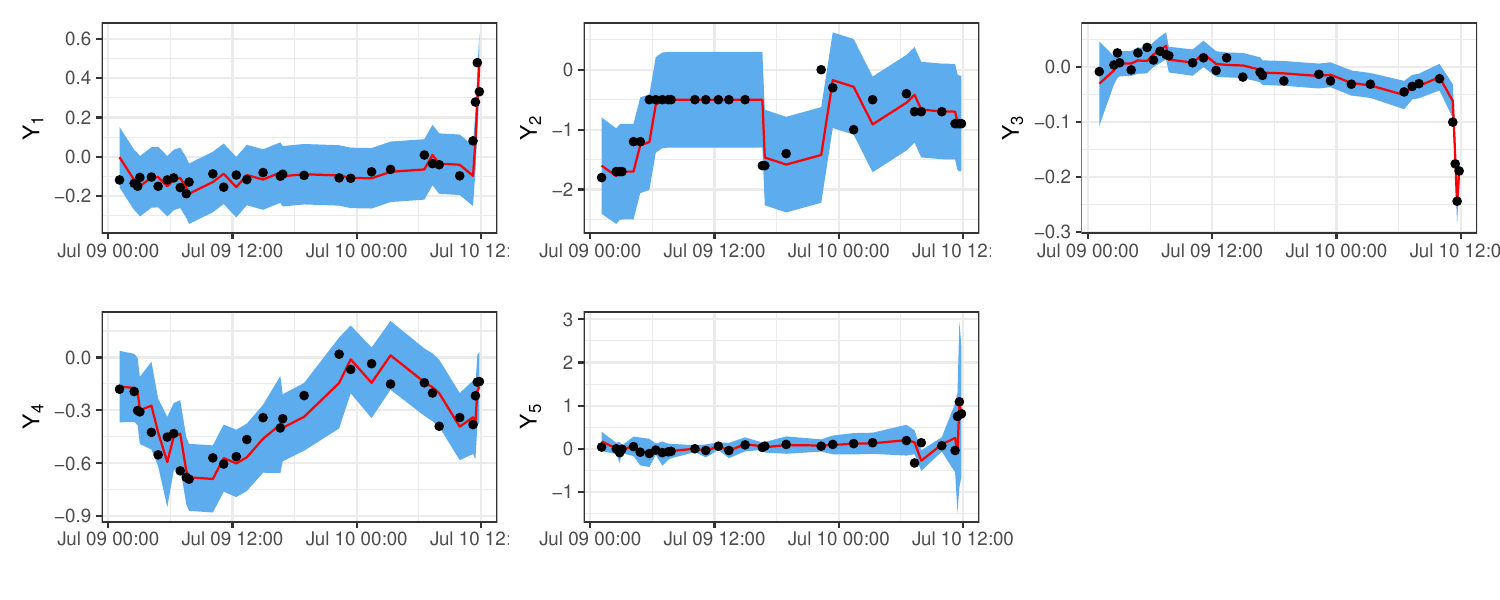}
\end{center}
\caption{Variables composing the CDSS and dynamic regression models' one-step-ahead predictions (mean and two standard deviation prediction intervals).}
\label{fig:algebraic_plot}
\end{figure}

The data were observed from time 1 to $T$, and clinicians have the choice to keep the patient on a low ventilator setting or increase the ventilator setting at $T+1$.
In Table \ref{tab:moments_clinical}, we list all the terms required to compute the expected utility function where expectations are conditioned on data derived by using Algorithm \ref{alg:cap} with Proposition \ref{th:polynomialder}. 

To compute the higher-order moments of the regression parameters in individual DLMs we can obtain these moments analytically using the moment-generating function (MGF) as they  are Multivariate Normal distributed.
Table \ref{tab:decision_clinician} provides probabilistic specifications that depend on the decision taken. In our analysis, we also specify equal criterion weights, $k_{1} = k_{2} = k_{3} = k_{d}=0.25$. Given these beliefs and parameter estimates, the expected utility score for this decision equals -0.0569, whereas for $d_{0}$ it is -0.465. Then the CDSS would suggest increasing the ventilator setting for the patient at time $T+1$. 

\begin{table}[h!]
\caption{Terms required to compute expected utility.}
\label{tab:moments_clinical}
 \begin{tabular}{|l l|} 
 \hline
$-k_1\E{\theta_{1t,0}}$ & $k_3\E{\theta_{3t,0}}$ \\ $k_3\E{\theta_{3t,1}}\E{\theta_{1t,0}}$ &
$-k_5\E{\theta_{5t,0000}}$ \\$-k_5\E{\theta_{5t,1000}}\E{\theta_{1t,0}}$ & 
$-k_5\E{\theta_{5t,0100}}
\E{\theta_{2t,0}}$ \\  $-k_5\E{\theta_{5t,0010}}\E{\theta_{3t,0}}$ &
$-k_5\E{\theta_{5t,0010}}\E{\theta_{3t,1}}\E{\theta_{1t,0}}$  \\
$-k_5\E{\theta_{5t,0001}}\E{\theta_{4t,000}}$ & $-k_5\E{\theta_{5t,0001}}\E{\theta_{4t,010}}\E{\theta_{2t,0}}$  \\
$-k_5\E{\theta_{5t,0001}}\E{\theta_{4t,001}}\E{\theta_{3t,0}}$  & $-k_5\E{\theta_{5t,0001}}\E{\theta_{4t,001}}\E{\theta_{3t,1}}\E{\theta_{1t,0}}$ \\ $-k_5\E{\theta_{5t,1010}}\E{\theta_{3t,0}}\E{\theta_{1t,0}}$  & 
$-k_5\E{\theta_{5t,1010}}\E{\theta_{3t,1}}\E{\theta_{1t,0}^2}$ \\
$-k_5\E{\theta_{5t,1010}}\E{\theta_{3t,1}}V_{1t}$ &
$-k_5\E{\theta_{5t,0101}}\E{\theta_{4t,000}}\E{\theta_{2t,0}}$ \\$-k_5\E{\theta_{5t,0101}}\E{\theta_{4t,010}}\E{\theta_{2t,0}^2}$  &  $-k_5\E{\theta_{5t,0101}}\E{\theta_{4t,010}}V_{2t}$ \\ $-k_5\E{\theta_{5t,0101}}\E{\theta_{4t,001}}\E{\theta_{3t,0}}\E{\theta_{2t,0}}$ 
& $-k_5\E{\theta_{5t,0101}}\E{\theta_{4t,001}}\E{\theta_{3t,1}}\E{\theta_{2t,0}}\E{\theta_{1t,0}}$ \\
$-k_5\E{\theta_{5t,0011}}\E{\theta_{4t,000}}\E{\theta_{3t,0}}$  & $-k_5\E{\theta_{5t,0011}}\E{\theta_{4t,000}}\E{\theta_{3t,1}}\E{\theta_{1t,0}}$  \\
$-k_5\E{\theta_{5t,0011}}\E{\theta_{4t,010}}\E{\theta_{3t,0}}\E{\theta_{2t,0}}$  &   $-k_5\E{\theta_{5t,0011}}\E{\theta_{4t,001}}\E{\theta_{3t,0}^2}$ \\
$-2k_5\E{\theta_{5t,0011}}\E{\theta_{4t,001}}\E{\theta_{3t,0}\theta_{3t,1}}\E{\theta_{1t,0}}$  & 
$-k_5\E{\theta_{5t,0011}}\E{\theta_{4t,001}}\E{\theta_{3t,1}^2}\E{\theta_{1t,0}^2}$  \\  $-k_5\E{\theta_{5t,0011}}\E{\theta_{4t,001}}\E{\theta_{3t,1}^2}V_{1t}$  & 
$-k_5\E{\theta_{5t,0011}}\E{\theta_{4t,001}}V_{3t}$ \\
$-k_5\E{\theta_{5t,0011}}\E{\theta_{4t,010}}\E{\theta_{3t,1}}\E{\theta_{2t,0}}\E{\theta_{1t,0}}$   & $- k_d\E{\rho(d)}$
\\ [1ex] 
 \hline
 \end{tabular}
\end{table}

\begin{table}[h!]
\caption{Probabilistic specifications that depend on the decision taken by clinicians}
\label{tab:decision_clinician}
\begin{center}
 \begin{tabular}{|c | c c c|} 
 \hline
& $\E{\theta_{1t,0}}$  & $\E{\theta_{3t,0}}$ & $\E{\rho(d)}$ \\
\hline
$d_{0}$ & 0.478 & -0.065 & 0.05\\
$d_{1}$ & 0.05 & 0.01 & 0.1
\\ [1ex] 
 \hline
 \end{tabular}
\end{center}
\end{table}

 \section{Concluding remarks}
\label{sec:discussion}
Probabilistic graphical models have been widely used to describe complex dependencies among variables. To perform inference, fast calculation of conditional and marginal probabilities on variables is required, motivating efficient algorithms for propagating information through the graph. We propose a method for a formal Bayesian decision analysis that combines polynomial regression models for individual variables within a graphical framework with moment propagation. The proposed approach scales up and captures complex relationships between individual variables and their parents. These conjectured relationships and the form of the utility allow the calculation of the expected utility scores by the decision-maker to discover optimal policies as a function of lower-order moments. This approach is distribution-free and does not require the full mixing associated with the full probability distribution of each model. We present a message-passing algorithm to allow efficient propagation of information for calculating scores and demonstrate how our approach can be used in decision-support in an energy planning domain and in a clinical setting.

Each vertex is a model developed by a panel of experts, and individually, no single panel can see what features are important in the final decision. Traditionally expert judgement at a panel may provide mean and variance, but as our network model results demonstrate, by coupling models in the graphical structure, it highlights that moments that are important to the decision maker often go beyond the mean and variance.The simplicity of our approach demonstrates the feasibility of the method, which can be extended to dynamic models in discrete time, where the complexity grows linearly. We give the exact evaluation of scores under specific assumptions, in which the methods can be used to approximate scores for different policy options. Of course, the sensitivity of these scores and such approximations would then need to be studied, a topic worthy of future research.

\section*{Acknowledgements}
The authors acknowledge Professor Henry Wynn (London School of Economics) for his valuable comments about the earlier versions of this work. The authors would like to thank Professor Chris Dent (University of Edinburgh) for his help with the energy planning example. The authors also acknowledge Dr. Samiran Ray (Great Ormond Street Hospital for Children NHS Foundation Trust) for providing data and expertise to construct clinical decision support system. 

This work was supported by the Additional Funding Programme for Mathematical Sciences, delivered by EPSRC (EP/V521917/1) and the Heilbronn Institute for Mathematical Research.

\bibliographystyle{plain}  
\bibliography{library}

\end{document}